\documentclass[12pt,a4paper,reqno]{amsart}

\title[Invariance Feedback Entropy of Uncertain Control Systems]{Invariance Feedback Entropy\\ of Uncertain Control Systems}

\author{Mahendra Singh Tomar}
\address{M. S. Tomar is with the Computer Science Department, Ludwig Maximilian University of Munich, Germany.}

\author{Matthias Rungger}
\address{M. Rungger is with the Hybrid Control Systems Group,
	Department of Electrical and Computer Engineering,
	Technical University of Munich, Germany.} 

\author{Majid Zamani}
\address{M. Zamani is with the Computer Science Department, University of Colorado Boulder, USA. M. Zamani is with the Computer Science Department, Ludwig Maximilian University of Munich, Germany.}
\email{mahendra.tomar@lmu.de,matthias.rungger@tum.de,majid.zamani@colorado.edu}

\usepackage[
bibencoding=latin1,
sorting=none,
giveninits=true,
style=numeric-comp,
maxnames=3,
backend=bibtex
]{biblatex}
\date{}

\usepackage{amsmath}
\usepackage{amsthm}
\usepackage{amssymb}

\usepackage{nicefrac}
\usepackage{mathtools}
\DeclarePairedDelimiter{\ceil}{\lceil}{\rceil}

\usepackage[usenames,dvipsnames,svgnames,table]{xcolor}

\usepackage{tikz}
\usetikzlibrary{calc,shapes,arrows}

\newtheorem{theorem}{Theorem}
\newtheorem{lemma}{Lemma}
\newtheorem{corollary}{Corollary}
\newtheorem{definition}{Definition}
\newtheorem{remark}{Remark}

\usepackage{thmtools}
\declaretheorem{example}
\renewcommand\thmcontinues[1]{Continued}
\DeclareMathOperator{\cl}{cl}
\newcommand{\no}{\raisebox{0.13\baselineskip}{\ensuremath{{\scriptstyle \#}}}}

\newcommand{\intcc}[1]{\ensuremath{{\left[#1\right]}}}
\newcommand{\intoc}[1]{\ensuremath{{\left]#1\right]}}}
\newcommand{\intco}[1]{\ensuremath{{\left[#1\right[}}}
\newcommand{\intoo}[1]{\ensuremath{{\left]#1\right[}}}

\newcommand{\R}{\mathbb{R}}
\newcommand{\N}{\mathbb{N}}
\newcommand{\Z}{\mathbb{Z}}

\newcommand{\E}{\mathbb{E}}
\newcommand{\id}{\mathrm{id}}

\renewcommand{\emptyset}{{\varnothing}}

\begin{document}

\begin{abstract}
We introduce a novel notion of invariance feedback entropy to quantify the
state information that is required by any controller that
enforces a given subset of the state space to be invariant. We establish a number
of elementary properties, e.g. we provide conditions that ensure that the
invariance feedback entropy is finite and show for the deterministic case that we recover the well-known notion of entropy for
deterministic control systems. We prove the data rate theorem, which shows that
the invariance entropy is a tight lower bound of the data rate of any
coder-controller that achieves invariance in the closed loop.
We analyze uncertain linear control systems and derive a universal lower bound
of the invariance feedback entropy. The lower bound depends on the absolute
value of the determinant of the system matrix and a ratio involving the volume
of the invariant set and the set of uncertainties.
Furthermore, we derive a lower bound of the
data rate of any static, memoryless coder-controller. Both lower bounds are
intimately related and for certain cases it is possible to bound the performance
loss due to the restriction to static coder-controllers by $1$ bit/time unit.
We provide various examples throughout the paper to illustrate and discuss
different definitions and results.
\end{abstract}

\maketitle
\section{Introduction}

In this work we study the classical feedback control loop, in which a
controller that is feedback connected with a given system is used to enforce a
prespecified control task in the closed loop. Unlike in
the classical setting, we do not assume that the sensor (or coder) is able to transmit an
infinite amount of information to the controller, but is restricted to use a
digital noiseless channel with a bounded data rate to communicate
with the controller. The closed loop of such a feedback is illustrated in
Fig. \ref{f:1}. In this context, we are interested in characterizing the
minimal data rate of the digital channel between coder and controller that
enables the controller to achieve the given control task. Or equivalently, we
are interested in quantifying the information required by the controller to
achieve a given control goal.

Data rate constrained feedback is a maturate research topic and has been
extensively studied for linear control systems and asymptotic stabilizability,
see e.g. \cite{NairFagniniZampieriEvans07} and references therein. Remarkably, for
this class of synthesis problems, the critical data rate has been
characterized in terms of the unstable eigenvalues of the system matrix independent
of the particular disturbance model \cite{hespanha2002towards,nair2003exponential,tatikonda2004control}.

We are interested in minimal data rates necessary
for a coder-controller scheme to render a given nonempty subset of the state
space invariant. Invariance specifications are one of the
most fundamental system requirements and are ubiquitous in the analysis and control
of dynamical systems~\cite{Aubin91,BlanchiniMiani08}. In
\cite{NairEvansMarrelsMoran04}, Nair et.~al extended the well-known notion of topological
entropy of dynamical
systems~\cite{adler1965topological,bowen1971entropy,dinaburg1971connection} to
discrete-time
deterministic control systems and showed that the topological feedback entropy
characterizes the data rate necessary to achieve invariance. Later Colonius
and Kawan~\cite{ColoniusKawan09} introduced a notion of invariance entropy for
continuous-time deterministic control systems. While the definition in
\cite{NairEvansMarrelsMoran04} clearly resembles the definition of entropy for
dynamical systems in \cite{adler1965topological} based on open covers, the
invariance entropy introduced in \cite{ColoniusKawan09} is close to the notion
of entropy in \cite{bowen1971entropy,dinaburg1971connection} based on spanning
sets. Both notions coincide for discrete-time control systems provided that a
strong invariance condition holds~\cite{ColoniusKawanNair13,Kawan13}.

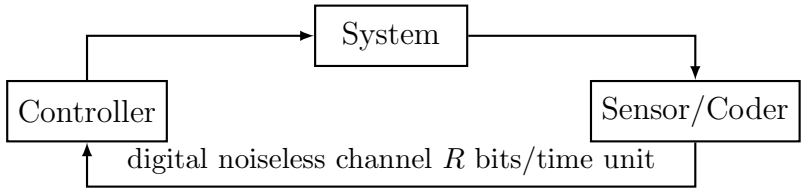
\begin{figure}[t]
 \centering
\begin{tikzpicture}[node distance=1.8cm, thick, >=latex]

\tikzset { 
	block/.style={ 
		draw,
		thick, 
		rectangle, 
		minimum height = .8cm, 
		minimum width = 2cm},
}
% first system
\draw node at (0,0) [block] (F)  {System};
\draw node at (4,-1) [block] (H)  {Sensor/Coder};
\draw node at (-4,-1) [block] (C)  {Controller};

%arrows
\draw[->] (F.east)  -|  (H.north);
\draw[->] (C.north)  |- (F.west);
\draw[->] (H.south)  |-       (0,-2)        node[above] {\small digital
	noiseless channel $R$ bits/time unit}  -|          (C.south);
\end{tikzpicture}
\caption{Coder-controller feedback loop.}\label{f:1}
\end{figure}

In this paper, we continue this line of research and introduce a notion of
invariance feedback entropy for uncertain control systems to characterize the
necessary state information required by any controller to enforce the
invariance condition in the closed loop. After we introduce the notation used in
this paper in Section \ref{s:notation}, we motivate the need of the novel
notion of invariance feedback entropy in Section \ref{s:motivation}.
We define
invariance feedback entropy and establish various elementary
properties in Section \ref{s:inv}. We show that the entropy is nonincreasing across two systems that
are related via a feedback refinement relation~\cite{reissig2016feedback}. This
result generalizes the fact that the invariance entropy of
deterministic control systems cannot increase under semiconjugation \cite[Thm 3.5]{ColoniusKawan09},
\cite[Prp. 2.13]{Kawan13}. We provide conditions that ensure that the invariance
feedback entropy is finite and show that we recover the notion of
invariance feedback entropy known for deterministic control systems, in the
deterministic case. We establish the data rate theorem in Section
\ref{s:data-rate}. It shows that the invariance entropy provides a tight lower
bound on the data rate of any coder-controller that enforces the invariance
specification in the closed loop. To this end, we introduce a history-dependent
notion of data rate. We discuss possible alternative data
rate definitions and motivate our particular choice by two examples.
We continue with the analysis of uncertain linear control systems in
Section~\ref{s:lin}. We derive a lower bound on the invariance feedback entropy.
The lower bound depends on the absolute value of the determinant of the system
matrix and a ratio involving the volume
of the invariant set and the set of uncertainties. The lower bound is invariant
under state space transformations and recovers the well-known minimal data
rate~\cite{NairFagniniZampieriEvans07} in the absence of uncertainties.
{\color{black} Similar to \cite[Section~II]{NairFagniniZampieriEvans07}, in the
	derivation we make use 
	of the Brunn-Minkowsky inequality for compact, measurable sets.}
Additionally, we derive a
lower bound of the data rate of any \emph{static}, \emph{memoryless}
coder-controller. Both lower bounds are intimately related 
and for certain cases it is possible to bound the performance
loss due to the restriction to static coder-controllers by 
$\log_2(1+1/2^{h_{\rm inv}{\color{black}(Q)}})$, where $h_{\rm inv}{\color{black}(Q)}$ is the invariance feedback
entropy of the uncertain linear systems, i.e., the best possible (dynamically)
achievable data rate, {\color{black}and $Q$ is the set of interest}. We show that the lower bounds are tight for certain
classes of systems.

A preliminary version of the results presented in Sections
\ref{s:motivation}-\ref{s:data-rate} appeared in~\cite{rungger2017invariance} and
the results on uncertain linear systems (Section \ref{s:lin}) appeared in \cite{rungger2017on}. This paper provides a detailed and
extended elaboration of the results proposed in
\cite{rungger2017invariance,rungger2017on}, including the new results  presented
in Theorem \ref{t:frr} and Theorem \ref{t:closed}.

\section{Notation}
\label{s:notation}

We denote by $\N$, $\Z$ and $\R$ the set of natural, integer and
real numbers, respectively. We annotate those symbols with subscripts to restrict the sets in
the obvious way, e.g. $\R_{>0}$ denotes the positive real numbers.
We denote the closed, open and half-open intervals in $\R$ with endpoints $a$ and $b$ by $\intcc{a,b}$,
$\intoo{a,b}$, $\intco{a,b}$, and $\intoc{a,b}$, respectively. The corresponding
intervals in $\Z$ are denoted by $\intcc{a;b}$, $\intoo{a;b}$, $\intco{a;b}$,
and $\intoc{a;b}$, i.e., $\intcc{a;b}=\intcc{a,b}\cap \Z$ and $\intco{a;a}=\emptyset$.

For a set $A$, we use $\no A\in \Z_{\ge0}\cup\{\infty\}$ to denote the number of
elements of $A$, i.e., if
$A$ is finite we have $\no A\in \Z_{\ge0}$ and  $\no A=\infty$ otherwise.
Given two sets $A$ and $B$, we say that $A$ is smaller (larger) than $B$ if $\no
A\le \no B$ ($\no A\ge \no B$) holds. 
A set $J$ of subsets of $A$ is said to \emph{cover} $B$, where $B\subseteq
A$, if $B$ is a subset of the union of {\color{black}elements of $J$ }.
A \emph{cover} of set $B$, is a set of subsets of $B$ that covers $B$.

{\color{black}We use $\exists_{a \in A}x =a$ to refer to: there exists $a$ in
	$A$ such that $x=a$. 
	In a similar way, $\forall_{a\in A}x=a$ is used.}
Given two sets $A,B\subseteq \R^n$, we define
the set addition by $A+B:=\{x\in\R^n\mid \exists_{a\in A
},\exists_{b\in B}\; x=a+b\}$ and denote by $\texttt{dim}(A)$ the dimension of set $A$. For $A=\{a\}$, we slightly abuse notation and use
$a+B=\{a\}+B$. The symbols $\cl A$, $\mathrm{int}A$ {\color{black}and $\wp(A)$} denote the closure,
the interior {\color{black}and the power set of a set} $A$, respectively.
We call a set $A\subseteq \R^n$ measurable if it is Lebesgue meas\-ur\-a\-ble and use $\mu(A)$
to denote its measure~\cite{tao2011introduction}. We use $\id$ to denote an identity map.

We follow \cite{RockafellarWets09} and use 
$f \colon A \rightrightarrows B$ to denote a \emph{set-valued map} from
$A$ into $B$, whereas $f \colon A \to B$ denotes an ordinary map.
If $f$ is set-valued, then $f$ is \emph{strict} if for every
$a\in A$ we have $f(a) \not= \emptyset$.
%\toremove{\color{black}For $D \subseteq B$ and set-valued $f$, by $f^{-1}(D)$ we refer to
%$\{ a\in A \mid f(a) \cap D \neq \emptyset  \}$.}
The restriction of $f$ to a subset $M \subseteq A$ is denoted by
$f|_{M}$. 
By convention we set $f|_\emptyset:=\emptyset$.  %and its domain by $\dom f=\{a\mid f(a)\neq \emptyset\}$.
The composition of $f:A\rightrightarrows B $ and
$g:C\rightrightarrows A$, $(f \circ g)(x) = f(g(x))$ is denoted by $f \circ g$. We use $B^A$
to denote the set of all functions $f:A\to B$.

%{\color{red} Strict is already explained.} {\color{brown}(Reviewer-1, comment-6 asked 'what is a strict relation'.)}  
%{\color{red} Yes, you can point him/her to the corresponding text above. No need to introduce it again.}
%{\color{black}We call a relation $R \subseteq A \times B$ strict if for every $a \in A$ we have $R(a) \neq \emptyset$.}
{\color{black}For a relation $R \subseteq A \times B$ and $D \subseteq A$, we define $R(D):=\cup_{d\in D} R(d).$}

The concatenation of two functions $x:\intco{0;a}\to X$ and
$y:\intco{0;b}\to X$ with $a\in \N$ and $b\in \N\cup\{\infty\}$
is denoted by $x y$ which we define by $x y(t):=x(t)$ for
$t\in\intco{0;a}$ and $x y(t):=y(t-a)$ for
$t\in\intco{a,a+b}$.

We use $\inf\emptyset=\infty$, $\log_2\infty=\infty$ and \mbox{$0
	\cdot \infty = 0$}.

\section{Motivation}
\label{s:motivation}
We study data rate constrained feedback for discrete-time \emph{uncertain
	control systems} described by 
difference inclusions of the form
\begin{equation}\label{e:ndi}
\xi(t+1)\in F(\xi(t),\nu(t))
\end{equation}
where $\xi(t)\in X$ is the \emph{state signal} and $\nu(t)\in U$ is the
\emph{input signal}. The sets $X$ and $U$ are referred to as \emph{state
	alphabet} and \emph{input alphabet}, respectively. The map $F:X\times
U\rightrightarrows X$ is called the
\emph{transition function}.

We are interested in coder-controllers that force the system
\eqref{e:ndi} to evolve inside a nonempty set $Q$ of the state alphabet $X$, i.e., every state signal $\xi$ of the
closed loop illustrated in Fig.~\ref{f:1} with $\xi(0)\in Q$ satisfies
$\xi(t)\in Q$ for all $t\in\Z_{\ge0}$. Specifically, we are interested in the
average data rate of such coder-controllers.

Notably, our system description is rather general and, depending on the structure of
alphabets  $X$ and $U$, we can represent a variety of commonly used system models. If we assume
$X$ and $U$ to be discrete, we can use~\eqref{e:ndi} to represent discrete event
systems\footnote{\small If \eqref{e:ndi} represents a discrete event system, the data
	rate unit is given in bits/event.}~\cite{CassandrasLafortune09} and digital/embedded
systems~\cite{BaierKatoen08}. Let us consider the following simple example.
\begin{example}\label{ex:intro}
	\normalfont
	Consider a system with state
	alphabet and input alphabet given by $X:=\{0,1,2\}$ and $U:=\{a,b\}$,
	respectively. The transition function is illustrated by
	
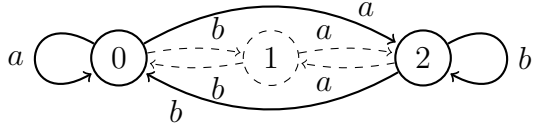
\begin{center} 
		\begin{tikzpicture}[->,thick,shorten >=1pt]
		\tikzstyle{state} = [draw, circle, minimum height=1.5em, minimum width=1.5em]
		
		\node[state] (0) at (-2,0) {$0$};
		\node[thin,dashed,state] (1) at (0,0) {$1$};
		\node[state] (2) at (2,0) {$2$};
		
		\path (0) edge [bend left] node[very near end,above] {$a$} (2);
		\path (2) edge [bend left] node[very near end,below] {$b$} (0);
		
		\path[thin,dashed] (0) edge[bend left=10] node[above,near end]   {$b$} (1);
		\path[thin,dashed] (1) edge[bend left=10] node[below,near start] {$b$} (0);
		
		\path[thin,dashed] (2) edge[bend left=10] node[below,near end] {$a$} (1);
		\path[thin,dashed] (1) edge[bend left=10] node[above,near start] {$a$} (2);
		
		\path (0) edge[in=210,out=150,loop]  node[left] {$a$} (0);
		\path (2) edge[out=30,in=-30,loop]  node[right] {$b$} (2);
		\end{tikzpicture}
	\end{center}
	The set of interest is defined to $Q:=\{0,2\}$.  The transitions and states that
	lead, respectively, are outside
	$Q$ are indicated by dashed lines.
	When the system is in state $0$ the only valid input is
	given by $a$. Similarly, if the system is in state $2$ the only valid input is given by $b$. If the input $a$ is
	applied at $0$ at time $t$, the system can either be in $0$ or $2$ at time $t+1$.
	Note that the valid control inputs for the states $0$ and $2$ differ and the controller
	is required to have exact state information at every point in time.
	Due to the nondeterministic transition function, 
	it is not possible to determine the current state of the system based on the knowledge of
	the past states, the past control inputs and the transition function. Therefore, the controller
	can obtain the state information only through measurement, which implies that at
	least one bit needs to be transmitted at every time step.% over the feedback channel.
	\qed
\end{example}
Current theories~\cite{NairEvansMarrelsMoran04,ColoniusKawan09,Kawan13,Colonius15}
are unable to explain the minimal data rate of one bit per time step observed in Example~\ref{ex:intro}.

If we allow $X$ and $U$ to be (subsets of) Euclidean spaces, we are able 
to recover one of the most fundamental system models in control theory, i.e.,
the class of nonlinear control systems with bounded
uncertainties~\cite{FreemanKokotovic96,BlanchiniMiani08}. If the system
description is given in continuous-time, we can use~\eqref{e:ndi} to represent
the sampled-data system~\cite{LailaNesicAstolfi06} with sampling time
$\tau\in\R_{>0}$ as illustrated in
Fig.~\ref{f:sd}.
\begin{figure}[h]
	\begin{center} 
		\begin{tikzpicture}[thick]
      \tikzset{block/.style = {draw, rectangle, minimum height = 3em, minimum width = 3em}}
\small

\node[block] (zoh) at (-2.5,0) {$\mathrm{ZOH}$};
\node[block] (sys) at (0.5,0) {$\dot \xi_c= f(\xi_c,\nu_c,\omega)$};

\draw[->,>=latex] (zoh.east) --
%node[above] {$\forall_{s\in \intco{t\tau ,(t+1)\tau}}$}
node[yshift=-1cm,xshift=-2.5mm] {$\forall_{s\in \intco{t\tau,(t+1)\tau}}\:\nu_c(s)=\nu(t)$} (sys.west);

\draw[<-,>=latex] (zoh.west) -- node[near end,above] {$\nu(t)$} ++(-1.25,0);

\draw (sys.east) -- (2.5,0);

\draw (2.5,0) -- node[above] {$\tau$}(2.9,.35);
\draw[dashed,<-] (2.75,0) --  (2.5,.35);
\draw[->,>=latex] (3,0) --  node[above,near end] {$\xi(t)$} ++(1.5,0);
\draw[fill] (2.5,0) circle (1pt);
\draw[fill] (3,0) circle (1pt);

\node at (2.75,-1cm) {$\xi_c(t\tau)=\xi(t)$};

\draw[dashed] (-3.75,-1.5) rectangle (3.75,.8);
		
		\end{tikzpicture}
		\caption{Sampled-data discrete-time system.}\label{f:sd}
	\end{center}
\end{figure}
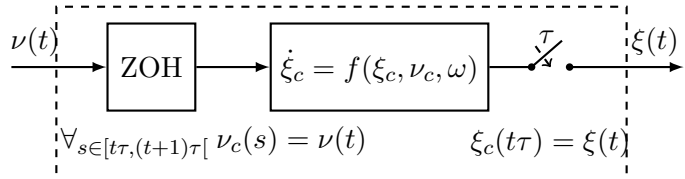
The disturbance signal $\omega$ is assumed to be
bounded $\omega(s)\in W\subseteq \R^p$ for all times $s\in\R_{\ge0}$.
The transition function $F(x,u)$ is defined as the set of states that are
reachable by the continuous-time system at time $\tau$ from initial state
$x$ under constant input signal $\nu_c(s)=u$ and a bounded disturbance signal
$\omega$. If the continuous-time dynamic is linear, the sampled-data system results in a
discrete-time system of the form
\begin{equation}\label{e:sys:lin}
\xi(t+1)\in A\xi(t)+ B\nu(t)+W
\end{equation}
where $A$ and $B$ are matrices of appropriate dimension  and $W$ is a nonempty
set representing the uncertainties.
%and $\xi(t)$ and $\nu(t)$ is the state, respectively, input signal.
%Here $x+W$ for $x\in\R^n$ denotes the Minkowsky set addition, which we define in
%the next section. 
%
\begin{example}\label{ex:linear}
	\normalfont
	Consider an instance of~\eqref{e:sys:lin} with $X:=\R$, $U:=\intcc{-1,1}$ and 
	\begin{equation*}
	F(x,u):=\tfrac{1}{2}x+u+\intcc{-3,3}
	\end{equation*}
	with the set of constraints given by $Q:=\intcc{-4,4}$.  \qed
\end{example}
For Example~\ref{ex:linear}, we establish in Section~\ref{s:lin}, that the smallest possible data rate of a
coder-controller that enforces $Q$ to be invariant is one bit per time step.
{\color{black} The example demonstrates that in contrast to linear systems without
	disturbances, where the data rate depends only on the unstable eigenvalues, see
	e.g.~\cite[Thm.~5.1]{ColoniusKawan09} or~\cite{tatikonda2004control}, for
	systems of the form \eqref{e:sys:lin} the
	data rate depends among other things also on the stable eigenvalues.}

%This is contrasts the results that are known for data rates of
%coder-controllers for controlled invariance for linear
%systems without disturbances and 
%
%This is in stark contrast to what is known for data rate constrained feedback
%control of linear systems with bounded disturbances in the context of asymptotic stabilization (or norm
%boundedness)~\cite[Thm.~1]{NairEvansMarrelsMoran04}, or for . 
%Both results suggest that the data rate
%should be zero, since the eigenvalue of the system matrix in Example~\ref{ex:linear} is given by
%$1/2$.

\section{Invariance Feedback Entropy}
\label{s:inv}

We introduce the notion of invariance feedback entropy and establish some
elementary properties. 

\subsection{The entropy}

Formally, we define a \emph{system} as triple 
\begin{equation}\label{e:sys}
\Sigma:=(X,U,F)
\end{equation}
where $X$ and $U$ are
nonempty sets and $F:X\times U\rightrightarrows X$ is assumed to be 
strict. 
A \emph{trajectory} of \eqref{e:sys} on $\intco{0;\tau}$ with $\tau\in\N\cup\{\infty\}$ is a pair
of sequences $(\xi,\nu)$, consisting of a state signal
\mbox{$\xi:\intco{0;\tau+1}\to X$} and an input signal $\nu:\intco{0;\tau}\to U$, that
satisfies~\eqref{e:ndi} for all $t\in\intco{0;\tau}$. We denote the set of all
trajectories on $\intco{0;\infty}$ by $\mathcal{B}(\Sigma)$.

Throughout the paper, we call a system  $(X,U,F)$
\emph{finite} if $X$ and $U$ are finite. 

We follow~\cite{NairEvansMarrelsMoran04} and~\cite[Sec.~6]{ColoniusKawan09} and define
the invariance feedback entropy with the help of covers of $Q$. 
Consider the system $\Sigma=(X,U,F)$ and a nonempty set
$Q\subseteq X$. A cover $\mathcal{A}$ of $Q$ and a function $G:\mathcal{A}\to U$ is called an
\emph{invariant cover}  $(\mathcal{A},G)$ of $\Sigma$ and $Q$ if $\mathcal{A}$ is finite and for all
$A\in \mathcal{A}$ we have $F(A,G(A))\subseteq Q$.

Consider an invariant cover $(\mathcal{A},G)$ of $\Sigma$ and $Q$, fix
$\tau\in\N$ and let $\mathcal{S}\subseteq
\mathcal{A}^{\intco{0;\tau}}$ be a set of sequences in $\mathcal{A}$. For
$\alpha\in \mathcal{S}$ and $t\in\intco{0;\tau-1}$ we define 
\begin{equation*}
P(\alpha|_{\intcc{0;t}}):=\{A\in\mathcal{A}\mid \exists_{\hat
	\alpha\in\mathcal{S}}\;\hat\alpha|_{\intcc{0;t}}=\alpha|_{\intcc{0;t}}\wedge
A=\hat\alpha(t+1)\}.
\end{equation*}
The set $P(\alpha|_{\intcc{0;t}})$ contains the cover elements $A$
so that the sequence $\alpha|_{\intcc{0;t}}A$ can be extended to a sequence in
$\mathcal{S}$.
For $t=\tau-1$ we have $\alpha|_{\intcc{0;\tau-1}}=\alpha$ and we define for
notational convenience the set
\begin{equation*}
P(\alpha):=\{A\in\mathcal{A}\mid \exists_{\hat \alpha\in\mathcal{S}}\;
A=\hat\alpha(0)\}
\end{equation*}
which is actually independent of $\alpha\in\mathcal{S}$ and corresponds to the
``initial'' cover elements $A$ in $\mathcal{S}$, i.e., there exists $\alpha\in
\mathcal{S}$ with $A=\alpha(0)$.
A set
$\mathcal{S}\subseteq\mathcal{A}^{\intco{0;\tau}}$ %with %$\tau\in\N$
is called
\emph{$(\tau,Q)$-spanning} in $(\mathcal{A},G)$ if the set $P(\alpha)$  with
$\alpha\in\mathcal{S}$ covers
$Q$ and we have
\begin{equation}\label{e:icover:spanning}
\forall_{\alpha\in\mathcal{S}}\forall_{t\in \intco{0;\tau-1}}\;
F(\alpha(t),G(\alpha(t)))\subseteq \bigcup_{A'\in
	P(\alpha|_{\intcc{0;t}})}A'.
\end{equation}
We associate with every $(\tau,Q)$-spanning set $\mathcal{S}$ the
\emph{expansion number} $N(\mathcal{S})$, which we define by 
\begin{equation*}
N(\mathcal{S}):=\max_{\alpha\in\mathcal{S}} \prod_{t=0}^{\tau-1}
\no P(\alpha|_{\intcc{0;t}}).
\end{equation*}
A tight lower bound on the expansion number of any $(\tau,Q)$-spanning set
$\mathcal{S}$ in $(\mathcal{A},G)$ is given by
\begin{equation*}
r_{\rm inv}(\tau,Q):=\min\left\{ N(\mathcal{S}) \mid \mathcal{S} \text{ is
	$(\tau,Q)$-spanning in } (\mathcal{A},G) \right\}.
\end{equation*}
We define the \emph{entropy} of an invariant cover $(\mathcal{A},G)$ by
\begin{equation}\label{e:icover:entropy}
h(\mathcal{A},G):=\lim_{\tau\to\infty}\frac{1}{\tau}\log_2 r_{\rm inv}(\tau,Q).
\end{equation}
As shown in Lemma~\ref{l:subadditivity} (stated below), the limit of the
sequence
in~\eqref{e:icover:entropy} exists so  that the
entropy of an invariant cover $(\mathcal{A},G)$ is well-defined.

The \emph{invariance feedback entropy} of $\Sigma$ and $Q$ follows by 
\begin{equation}\label{e:entropy}
h_{\rm inv}{\color{black}(Q)}:=\inf_{(\mathcal{A},G)}h(\mathcal{A},G)
\end{equation}
where we take the infimum over all $(\mathcal{A},G)$ invariant covers of $\Sigma$ and
$Q$. Let us revisit the examples from the previous section to illustrate the
various definitions. 

\begin{example}[continues=ex:intro]
	\normalfont
	First, we determine an invariant cover $(\mathcal{A},G)$ of the system in
	Example~\ref{ex:intro} and $Q$. Since the system is
	finite, we can set $\mathcal{A}:=\{\{x\}\mid x\in Q\}$. Recall that $Q=\{0,2\}$ and a suitable
	function $G$ is given by $G(\{0\}):=a$ and $G(\{2\}):=b$. Suppose that 
	$\mathcal{S}\subseteq \mathcal{A}^{\intco{0;\tau}}$ is $(\tau,Q)$-spanning
	with $\tau\in\N$. Let us look at condition~\eqref{e:icover:spanning} for $t\in\intco{0;\tau-1}$ and
	$\alpha\in\mathcal{S}$. If $\alpha(t)=\{0\}$, we have
	$P(\alpha {\color{black}|_\intcc{0;t}})=\{\{0\},\{2\}\}$ since $F(\{0\},G(\{0\}))=F(0,a)=\{0,2\}$. If
	$\alpha(t)=\{2\}$ the same reasoning leads to $P(\alpha{\color{black}|_\intcc{0;t}})=\{\{0\},\{2\}\}$.
	Also for $\alpha\in\mathcal{S}$ we have $P(\alpha)=\{\{0\},\{2\}\}$ since
	$P(\alpha)$ is required to be a cover of $Q$. It follows that
	$\mathcal{S}=\mathcal{A}^{\intco{0;\tau}}$ and the expansion number
	$N(\mathcal{S})=r_{\rm inv}(\mathcal{A},G)=2^\tau$ so that the entropy of the
	$(\mathcal{A},G)$ follows to $h(\mathcal{A},G)=1$. Since
	$(\mathcal{A},G)$ is the only invariant cover we obtain $h_{\rm inv}{\color{black}(Q)}=1$.
	\qed
\end{example}
\begin{example}[continues=ex:linear]
	\normalfont
	Let us recall the linear system in Example~\ref{ex:linear}. An invariant cover
	$(\mathcal{A},G)$ is given by $\mathcal{A}:=\{a_0,a_1\}$ with
	$a_0:=\intcc{-4,0}$, $a_1:=\intcc{0,4}$ and $G(a_0):=1$, $G(a_1):=-1$. 
	{\color{black}Let $\mathcal{S}$ be any $(\tau,Q)$-spanning set in $(\mathcal{A},G)$.
		As $P(\alpha)\subseteq \mathcal{A}$ is required to cover $Q$, so $P(\alpha) = \mathcal{A}$.
		For $a_i\in\mathcal{A}$, $i\in\{0,1\}$ we have $F(a_i,G(a_i)) = \intcc{-4;4}$ which makes $P(a_i)=\mathcal{A}$. Thus $\mathcal{S} = \mathcal{A}^\intco{0;\tau}$.       
		%  We use a similar reasoning as in
		%  Example~\ref{ex:intro} to see that for every $\tau\in\N$ the only
		%  $(\tau,Q)$-spanning set is $\mathcal{S}:=\mathcal{A}^{\intco{0;\tau}}$. 
	}
	Since $\no \mathcal{A}=2$, we obtain that
	$h(\mathcal{A},G)=1$.\qed
\end{example}
We continue with showing the subadditivity of $\log_2 r_{\rm inv}(\cdot,Q)$.
\begin{lemma} \label{l:subadditivity}
	Consider the system $\Sigma=(X,U,F)$ and a non\-empty set $Q\subseteq X$.
	Let $(\mathcal{A},G)$ be an invariant cover of $\Sigma$ and $Q$,
	then the function
	$\tau\mapsto \log_2 r_{\rm inv}(\tau,Q)$, $\N\to \R_{\ge0}$ is subadditive, i.e., 
	for all $\tau_1,\tau_2\in\N$ the inequality
	$$
	\log_2 r_{\rm inv}(\tau_1+\tau_2,Q)
	\le
	\log_2 r_{\rm inv}(\tau_1,Q) 
	+ 
	\log_2 r_{\rm inv}(\tau_2,Q)
	$$ holds and we have
	\begin{equation}\label{e:ifb:equiv}
	\lim_{\tau\to\infty}\frac{1}{\tau}\log_2 r_{\rm inv}(\tau,Q)
	=
	\inf_{\tau\in\N}\frac{1}{\tau}\log_2 r_{\rm inv}(\tau,Q).
	\end{equation}
\end{lemma}

The following lemma might be of independent interest. 
We use it in the proof of Theorem~\ref{t:invdet}.
\begin{lemma}\label{l:nospanningset}
	Consider an invariant cover $(\mathcal{A},G)$ of~\eqref{e:sys} and some
	nonempty set $Q\subseteq X$. Let $\mathcal{S}$ be a $(\tau,Q)$-spanning set,
	then we have $\no \mathcal{S}\le N(\mathcal{S})$.
\end{lemma}
The proofs of both lemmas are given in the appendix.

\subsection{Entropy across related systems}

One of the most important properties of entropy of classical dynamical systems
is its invariance under any change of
coordinates~\cite[Thm.~1]{adler1965topological}. In~\cite{ColoniusKawan09}
this property has been shown for deterministic control systems in the context of
semiconjugation~\cite[Thm.~3.5]{ColoniusKawan09}.
In the following, we present a result in the
context of feedback refinement relations~\cite{reissig2016feedback}, which
contains the result on semiconjugation as a special case.

\begin{definition}\label{d:frr}
	Let $\Sigma_1$ and $\Sigma_2$ be two systems of the form
	\begin{equation} \label{e:systems}
	\Sigma_i=(X_i, U_i, F_i) \text{ with } i\in\{1,2\}.
	\end{equation}
	A strict relation $R \subseteq X_1 \times X_2$ is a
	\emph{feedback refinement relation} from $\Sigma_1$ to $\Sigma_2$ if there exists a map
	$r:U_2\to U_1$ so that the following inclusion holds for all $(x_1,x_2) \in {\color{black}R}$
	and $u\in U_2$
	\begin{equation} \label{e:frr}
	R(F_1(x_1,r(u))) \subseteq F_2(x_2, u).
	\end{equation}
\end{definition}

\begin{theorem}\label{t:frr}
	Consider two systems $\Sigma_i$, $i\in\{1,2\}$ of the form~\eqref{e:systems}.
	Let $Q_1$ and $Q_2$ be two nonempty subsets of $X_1$ and $X_2$, respectively.
	Suppose that $R$ is a feedback refinement relation from $\Sigma_1$ to $\Sigma_2$,
	and $Q_1=R^{-1}(Q_2)$.
	%{\color{black}and $R^{-1}:Q_2\rightrightarrows Q_1$ is strict. }
	Then
	\begin{equation} \label{e:frr:ineq}
	h_{1,\rm inv}{\color{black}(Q_1)}\le h_{2,\rm inv}{\color{black}(Q_2)}
	\end{equation}
	holds, where $h_{i,\rm inv}{\color{black}(Q_i)}$ is the invariance feedback entropy of $\Sigma_i$
	and $Q_i$. 
\end{theorem}
\begin{proof}
	{\color{black}
		If $h_{2,\rm inv}(Q_2)=\infty$, the inequality holds and subsequently we consider the
		case $h_{2,\rm inv}(Q_2)<\infty$. 
		We will make use of Lemma~\ref{lem:MPG} in the Appendix to show~\eqref{e:frr:ineq}.
		Let us pick an	invariant cover $(\mathcal{A}_2,G_2)$  of 
		$\Sigma_2$ and $Q_2$ so that
		$h(\mathcal{A}_2,G_2)<\infty$. 
		Next we define the set $\mathcal{A}_1:=\{A_1\subseteq Q_1\mid
		\exists_{A_2\in\mathcal{A}_2}\;R^{-1}(A_2)=A_1\}$.
		
		Now let $M=R^{-1}$ and $r:U_2\to U_1$ in Lemma~\ref{lem:MPG}, where $R$ and $r$ are, respectively, the relation and map associated with the feedback refinement relation in Def.~\ref{d:frr}. 
		We observe that all the 
		conditions \ref{lem:MPG:cover}) - \ref{lem:MPG:semconj}) in Lemma~\ref{lem:MPG} hold.
		
		%{\color{olive}
		%\begin{enumerate}
		%	\item $R^{-1}(Q_2) = Q_1$ , (assumption)
		%	
		%	\item Let $A_2, A_2', B \subseteq X_2$ such that $A_2'=A_2\cup B$. 
		%	From definition of $R^{-1}$
		%	\begin{align*}
		%	  	R^{-1}(A_2') &= \left\{ x\in X_1 \mid \exists_{y\in A_2'}(x,y)\in R\right\} \\
		%	  	&= \left\{ x\in X_1 \mid \exists_{y\in A_2}(x,y)\in R\right\} \\ 
		%	  	& \quad \cup \left\{ x\in X_1 \mid \exists_{y\in B}(x,y)\in R\right\} \\
		%	  	&= R^{-1}(A_2) \cup R^{-1}(B)
		%	\end{align*}
		%	Thus $A_2\subseteq A_2'\implies R^{-1}(A_2)\subseteq R^{-1}(A_2')$.
		%	
		%	\item From 2), $R^{-1}(A_2\cup B) = R^{-1}(A_2) \cup R^{-1}(B)$
		%	
		%	\item Given $R(F_1(x_1,r(u))) \subseteq F_2(x_2, u)$ for $(x_1,x_2)\in R$. So
		%	\begin{align*}
		%  	R(F_1(R^{-1}(A_2),r(u))) &\subseteq F_2(A_2, u) \\
		%  \implies	F_1(R^{-1}(A_2),r(u)) &\subseteq R^{-1}(F_2(A_2, u)).
		%	\end{align*}
		% % 	$A_1:= R^{-1}(A_2)$ so for every $x\in A_1$ there exists $y\in A_2$ such that
		% % 	$(x,y)\in R$. Given 
		%\end{enumerate} }
		
		Thus there exists a map $G_1^*:\mathcal{A}_1\to U_1$ such that $(\mathcal{A}_1,G_1^*)$ is an invariant cover of $\Sigma_1$ and $Q_1$, and
		\begin{equation*}
		h(\mathcal{A}_1,G_1^*) \leq h(\mathcal{A}_2,G_2).
		\end{equation*}
		Therefore, inequality~\eqref{e:frr:ineq} holds. 
	}
\end{proof}

\subsection{Conditions for finiteness}

We analyze two particular instances of systems -- finite
systems and systems with a topological state alphabet  -- and provide conditions ensuring that the invariance
entropy is finite. The results are based on the following lemma.
\begin{lemma}\label{l:finite}
	Consider a system $\Sigma=(X,U,F)$ and a non\-empty set $Q\subseteq X$.
	There exists an invariant cover
	$(\mathcal{A},G)$ of $\Sigma$ and $Q$ iff $h_{\rm inv}{\color{black}(Q)}<\infty$.
\end{lemma}
\begin{proof}
	It follows immediately from~\eqref{e:entropy} that $h_{\rm
		inv}{\color{black}(Q)}<\infty$ implies the existence of an invariant cover of $\Sigma$ and $Q$.
	For the reverse direction, we assume that $(\mathcal{A},G)$ is an invariant
	cover of $\Sigma$ and $Q$.
	We fix $\tau\in \N$ and define
	$\mathcal{S}:=\{\alpha\in\mathcal{A}^{\intco{0;\tau}}\mid
	\forall_{t\in\intco{0;\tau-1}}\;\alpha(t+1)\cap
	F(\alpha(t),G(\alpha(t)))\neq\emptyset \}$. It is easy to verify that
	$\mathcal{S}$ is $(\tau,Q)$-spanning and $N(\mathcal{S})\le
	(\no\mathcal{A})^{\tau}$. An upper bound on $h_{\rm inv}{\color{black}(Q)}$ follows
	by $\log_2\no\mathcal{A}$.
\end{proof}

If $\Sigma$ is finite, it is rather straightforward to show that the controlled invariance of
$Q$ w.r.t.~$\Sigma$ is necessary and sufficient for $h_{\rm inv}{\color{black}(Q)}$ to be finite.
Let us recall the notion of controlled invariance~\cite{BlanchiniMiani08}.

We call $Q\subseteq X$ \emph{controlled invariant} with respect to a system
$\Sigma=(X,U,F)$, if for all $x\in Q$ there exists $u\in U$ so that $F(x,u)\subseteq Q$. 
{\color{black} We refer the interested readers to \cite{rungger2017computing} for a discussion on computation of controlled invariant set for 
	controllable linear discrete-time systems.}

\begin{theorem}\label{t:finite:finite}
	Consider a finite system $\Sigma=(X,U,F)$ and a nonempty set $Q\subseteq X$.
	Then $h_{\rm inv}{\color{black}(Q)}<\infty$ if and only if $Q$ is controlled invariant.
\end{theorem}

\begin{proof}
	Let $h_{\rm inv}{\color{black}(Q)}$ be finite. Then there exists an invariant cover
	$(\mathcal{A},G)$ so that $h(\mathcal{A},G)<\infty$. Hence, for every $x\in Q$
	we can pick an $A\in\mathcal{A}$ with $x\in A$, so that $F(x,G(A))\subseteq
	F(A,G(A))\subseteq Q$. Hence $Q$ is controlled invariant w.r.t.~$\Sigma$.
	
	Assume $Q$ is controlled invariant w.r.t.~$\Sigma$. For $x\in Q$, let $u_x\in U$ be such
	that $F(x,u_x)\subseteq Q$. It is easy to check that $(\mathcal{A},G)$ with
	$\mathcal{A}:=\{\{x\}\mid x\in Q\}$ and $G(\{x\}):=u_x$ is an invariant cover
	of $\Sigma$ and $Q$, so that the assertion follows from Lemma~\ref{l:finite}.
\end{proof}

In general controlled invariance
of $Q$ is not sufficient to guarantee finiteness of the invariance feedback
entropy as shown in the next example.
%In particular, the finiteness of $\Sigma$ in the previous
%Theorem~\ref{t:finite:finite} was essential.
\begin{example}[label=ex:notfinite]
	\normalfont
	Consider $\Sigma=(\R,\intcc{-1,1},F)$ with the dynamics given by
	$F(x,u):=x+u+\intcc{-1,1}$. Let $Q:=\intcc{-1,1}$, then for every
	$x\in Q$ we can pick $u=-x$ so that $F(x,u)=\intcc{-1,1}\subseteq Q$, which
	shows that $Q$ is controlled invariant.
	Now suppose that $h_{\rm inv}{\color{black}(Q)}$ is finite. Then
	according to Lemma~\ref{l:finite} there exists an invariant cover
	$(\mathcal{A},G)$ of $\Sigma$ and $Q$. Since $\mathcal{A}$ is required to be
	finite, there exists $A\in\mathcal{A}$ with an infinite number of elements and
	therefore we can pick two different states in $A$, i.e., $x,x'\in A$ with
	$x\neq x'$. However, there does not exist a single $u\in U$ so that
	$F(x,u)\subseteq Q$ and $F(x',u)\subseteq Q$. Hence, $(\mathcal{A},G)$ cannot
	be an invariant cover, which implies $h_{\rm inv}{\color{black}(Q)}=\infty$.
	\qed
\end{example}

In the subsequent theorem we present some conditions for systems with a  topological state alphabet,
which imply the finiteness of the invariance entropy. {\color{black} The
	conditions may be difficult to verify for a particular problem instance.}
Nevertheless with these conditions, we
follow closely the assumptions based on continuity and strong invariance employed
in~\cite{NairFagniniZampieriEvans07, ColoniusKawanNair13} to ensure finiteness
of the invariance entropy for deterministic systems. We use the following notion of
continuity of set-valued maps~\cite{AubinFrankowska90} to show the next result.

Let $A$ and $B$ be topological
spaces and $f:A\rightrightarrows B$. We say that $f$ is \emph{upper
	semicontinuous}, if for every $a\in A$ and every open set
$V\subseteq B$ containing $f(a)$ there exists an open set $U\subseteq A$ with $a\in U$ so that
$f(U)\subseteq V$.

\begin{theorem}\label{t:finite:topological}
	Consider a system $\Sigma=(X,U,F)$ and a nonempty compact subset $Q$  of $X$.
	Let $X$ be a topological space.  If $F(\cdot,u)$ is upper semicontinuous for
	every $u\in U$ and $Q$ is strongly controlled invariant, i.e., for all $x\in
	Q$ there exists $u\in U$ so that $F(x,u)\subseteq \mathrm{int}\:Q$, then
	$h_{\rm inv}{\color{black}(Q)}<\infty$.
\end{theorem}
\begin{proof}
	For each $x\in Q$, we pick an input $u_x\in U$ so that $F(x,u_x)\subseteq
	\mathrm{int}\: Q$. Since $F(\cdot,u_x)$ is  upper semicontinuous and
	$\mathrm{int}\:Q$ is open, there exists an open subset $A_x$ of $X$, so that $x\in A_x$ and $F(A_x,u_x)\subseteq \mathrm{int}\:Q$. Hence, the
	set $\{A_x\mid x\in Q\}$ of open subsets of $X$ covers $Q$. Since $Q$ is a
	compact subset of $X$, there exists a finite set $\{A_{x_1},\ldots,
	A_{x_m}\}$ so that $Q\subseteq \cup_{i\in\intcc{1;m}}
	A_{x_i}$~\cite[Ch.~2.6]{GamelinGreen99}. Let
	$\mathcal{A}:=\{A_{x_1}\cap Q,\ldots, A_{x_m}\cap Q\}$ and define for every
	$i\in\intcc{1;m}$ the function $G(A_{x_i}):=u_{x_i}$.  Then
	$(\mathcal{A},G)$ is an invariant cover of $\Sigma$ and $Q$, and the assertion
	follows from Lemma~\ref{l:finite}.
\end{proof}

\begin{example}[continues=ex:notfinite]
	\normalfont
	Let $\varepsilon>0$, consider $\Sigma$ from Example~\ref{ex:notfinite} with the modified input set
	$U_\varepsilon:=\intcc{-1-\varepsilon,1+\varepsilon}$. Let
	$Q_\varepsilon:=\intcc{-1-\varepsilon,1+\varepsilon}$ then we see that $Q_\varepsilon$ is strongly controlled invariant.  We construct an invariant cover for $\Sigma$ and $Q_\varepsilon$ as follows.
	We define $n$ as the smallest integer larger than $\tfrac{1}{2\varepsilon}$ and introduce
	$\{x_{-n},\ldots,x_0,\ldots x_{n}\}$ with $x_i:=2i\varepsilon$ and set
	$A_i:=(x_i+\intcc{-\varepsilon,\varepsilon})\cap Q_\varepsilon$. For each
	$i\in\intcc{-n;n}$ we define $G(A_i):=-x_i$ so that
	$F(A_i,G(A_i)) {\color{black}\subseteq} Q_\varepsilon$.
	By definition of $n$ we have  $x_{-n}\le -1$ and $1\le x_n$ and we see that
	$(\mathcal{A},G)$ with $\mathcal{A}:=\{A_i\mid i\in\intcc{-n;n}\}$ is an
	invariant cover of $\Sigma$ and $Q_\varepsilon$. Hence, it follows from
	Lemma~\ref{l:finite} that $h_{\rm inv}{\color{black}(Q_\varepsilon)}$ is finite.~\qed
\end{example}

\subsection{Deterministic systems}

For deterministic systems we recover the notion of invariance feedback entropy
in~\cite{NairEvansMarrelsMoran04,ColoniusKawanNair13}.

%This
%fact follows indirectly from the data rate theorem (Theorem~\ref{t:datarate}).
%Nevertheless, we
%provide here a direct proof, which demonstrates in greater detail the
%relations of the various objects in the definition of invariance feedback
%entropy for nondeterministic and deterministic systems.

Let us consider the map $f:X\times U\to X$ representing a deterministic system 
\begin{equation}\label{e:dt:sys}
\xi(t+1) = f(\xi(t),\nu(t)).
\end{equation}
We can interpret~\eqref{e:dt:sys} as special instance
of~\eqref{e:sys}, where $F$ is given by $F(x,u):=\{f(x,u)\}$ for all $x\in X$
and $u\in U$ and the notions of a trajectory of~\eqref{e:sys} extend
to~\eqref{e:dt:sys} in the obvious way. 
Given an input $u\in U$, we introduce $f_u: X\to X$ by $f_{u}(x):=f(x,u)$
and extend this notation to sequences $\nu\in U^{\intcc{0;t}}$, $t\in\N$
by
\begin{equation*}
f_{\nu}(x):=f_{\nu(t)}\circ\cdots\circ f_{\nu(0)}(x).
\end{equation*}

We follow~\cite{ColoniusKawanNair13} to define the entropy
of~\eqref{e:dt:sys}.
Consider a nonempty set $Q\subseteq X$ and fix $\tau\in \N$. A
set $\mathcal{S}\subseteq U^{\intco{0;\tau}}$ is called
\emph{$(\tau,Q)$-spanning} for $f$ and $Q$, if for every $x\in Q$ there exists $\nu\in
\mathcal{S}$ so that the associated trajectory $(\xi,\nu)$  on $\intco{0;\tau}$ of~\eqref{e:dt:sys} with $\xi(0)=x$ satisfies $\xi(\intcc{0;\tau})\subseteq Q$.
We use $r_{\rm det}(\tau,Q)$ to
denote the number of elements of the smallest $(\tau,Q)$-spanning set 
\begin{equation}\label{e:dt:min:span}
r_{\rm det}(\tau,Q):= \inf\{ \no\mathcal{S}\mid \mathcal{S} \text{ is $(\tau,Q)$-spanning}\}.
\end{equation}
The \emph{(deterministic) invariance entropy} of $(X,U,f)$ and $Q$
is defined by
\begin{equation}\label{e:dt:inv}
h_{\rm det}{\color{black}(Q)}:=\lim_{\tau\to\infty}\frac{1}{\tau}\log_2 r_{\rm det}(\tau,Q).
\end{equation}
Again the function $\tau\mapsto\frac{1}{\tau}\log_2 r_{\rm det}(\tau,Q)$ is
subadditive~\cite[Prop.~2.2]{ColoniusKawanNair13} which ensures that the limit
in~\eqref{e:dt:inv} exists.

Now, we have the following theorem.
\begin{theorem}\label{t:invdet}
	Consider the system $\Sigma=(X,U,F)$ and a nonempty set
	$Q\subseteq X$. 
	Suppose $F$ satisfy $F(x,u)=\{f(x,u)\}$ for all $x\in X$, $u\in
	U$ for some $f:X\times U\to X$. Then the
	invariance feedback entropy of $\Sigma$ and $Q$ equals the deterministic
	invariance entropy of $(X,U,f)$ and $Q$, i.e.,
	\begin{equation}%\label{e:dt:equiv}
	h_{\rm inv}{\color{black}(Q)}=h_{\rm det}{\color{black}(Q)}.
	\end{equation}
\end{theorem}
\begin{proof}
	We begin with the inequality $h_{\rm det}{\color{black}(Q)}\ge h_{\rm inv}{\color{black}(Q)}$. If $h_{\rm det}{\color{black}(Q)}=\infty$ the
	inequality trivially holds and subsequently we assume that $h_{\rm det}{\color{black}(Q)}$ is
	finite.
	%which implies that $r_{\rm det}(\tau,Q)<\infty$ for all $\tau\in \N$,
	%see~\cite[Rem.~2.1]{Kawan13}.
	We fix $\varepsilon>0$ and pick $\tau\in \N$ so
	that $\frac{1}{\tau}\log_2 r_{\rm det}(\tau,Q)\le h_{\rm det}{\color{black}(Q)}+\varepsilon$. We
	chose a $(\tau,Q)$-spanning set $\mathcal{S}_{\rm det}$ for $f$ and $Q$ with $\no
	\mathcal{S}_{\rm det}=r_{\rm det}(\tau,Q)$. For every $\nu\in\mathcal{S}_{\rm det}$ 
	we define the sets
	\begin{equation*}
	A_0(\nu):=Q\cap\bigcap_{t=0}^{\tau-1}f_{\nu|_{\intcc{0;t}}}^{-1}(Q)
	\end{equation*}
	and for $t\in\intco{0;\tau-1}$ the sets
	$A_{t+1}(\nu):=f(A_t(\nu),\nu(t))$. The minimality of 
	$\mathcal{S}_{\rm det}$ implies that $A_0(\nu)\neq \emptyset$ and $A_0(\nu)\neq A_0(\nu')$ 
	for all $\nu,\nu'\in\mathcal{S}_{\rm det}$.
	Let $\mathcal{A}$ be the set
	of all sets $A_{t}(\nu)$ with $t \in\intco{0;\tau}$ and $\nu
	\in\mathcal{S}_{\rm det}$.
	\begin{color}{black}
		With each $A \in \mathcal{A}$ we associate a single pair $(\nu,t)$, where
		$\nu\in\mathcal{S}_{\rm det}$ and $t\in\intco{0;\tau}$, such that is satisfies
		$A=A_t(\nu)$ and the
		following condition: $\nu'\in\mathcal{S}_{\rm det}$ and
		$t'\in\intco{0;\tau}$ with $A=A_{t'}(\nu')$ implies $t\le t'$. Then we define
		the map $G:\mathcal{A}\to U$ by 
		$G(A)=\nu(t)$ where $(\nu,t)$ is associated with $A$.
	\end{color}
	By the definition of $A_t(\nu)$, it is easy to see that 
	$f(A_t(\nu),G(A_t(\nu)))\subseteq Q$ for all $t\in\intco{0;\tau}$ and
	{$\nu\in\mathcal{S}_{\rm det}$}. Moreover, since $\mathcal{S}_{\rm det}$ is $(\tau,Q)$-spanning, for
	every $x\in Q$ there is $\nu\in\mathcal{S}_{\rm det}$ so that for all
	$t\in\intco{0;\tau}$ we have $f_{\nu|_{\intcc{0;t}}}(x)\in Q$ which implies 
	$x\in A_0(\nu)$ and we see that $\{A_0(\nu)\mid \nu\in\mathcal{S}_{\rm det}\}$
	covers $Q$.  It follows that $(\mathcal{A},G)$ is an
	invariant cover of $(X,U,F)$ and $Q$.  
	{\color{black} Let $\mathcal{S}_{\text{inv}}$ be the set of sequences 
		$\alpha:\intco{0;\tau}\to \mathcal{A}$ defined iteratively as 
		$\alpha(0) \in \{A_0(\nu) \mid \nu \in \mathcal{S}_{det} \}$ and
		$\alpha(t+1) = f(\alpha(t), G(\alpha(t)))$. 
		Then $P(\alpha)$ covers $Q$ since 
		$\{A_0(\nu) \mid \nu \in \mathcal{S}_{det} \}$ covers $Q$ as discussed above.
		For any distinct $\alpha, \alpha' \in \mathcal{S}_{\text{inv}}$ we have
		$\alpha(0) \neq \alpha'(0)$ so for
		every $t \in \intco{0; \tau -1}$ we have 
		$\no P(\alpha|_{\intcc{0;t}})=1$, 
		$f(\alpha(t),G(\alpha(t)))= P(\alpha|_{\intcc{0;t}}) $ and thus
		$\mathcal{S}_{\text{inv}}$ satisfies~\eqref{e:icover:spanning}.
	}
	Therefore,
	$\mathcal{S}_{\rm inv}$ is $(\tau,Q)$-spanning in $(\mathcal{A},G)$.  
	Moreover, as $\nu\neq\nu'$ implies 
	{\color{black}$A_0(\nu)\neq A_0(\nu')$}, we have
	$\no P(\alpha)= \no \mathcal{S}_{\rm det}$, so that
	$r_{\rm inv}(\tau,Q)\le N(\mathcal{S}_{\rm inv})= \no
	\mathcal{S}_{\rm det}=r_{\rm det}(\tau,Q)$ follows. Due to
	Lemma~\ref{l:subadditivity}, we have $\log_2 r_{\rm inv}(n\tau,Q)\le n \log_2 r_{\rm
		inv}(\tau,Q)$ and we see that $\frac{1}{\tau}\log_2 r_{\rm inv}(\tau,Q)$ (and
	therefore $\frac{1}{\tau}\log_2 r_{\rm det}(\tau,Q)$) provides an
	upper bound for $h_{\rm}(\mathcal{A},G)$ so that 
	we obtain $h_{\rm inv}{\color{black}(Q)}\le $ $h_{\rm}(\mathcal{A},G)$ $ \le h_{\rm det}{\color{black}(Q)}+\varepsilon$.
	Since this holds for any $\varepsilon>0$ we obtain the desired inequality.\\
	We continue with the inequality $h_{\rm det}{\color{black}(Q)}\le h_{\rm inv}{\color{black}(Q)}$. 
	If $h_{\rm inv}{\color{black}(Q)}=\infty$ the
	inequality trivially holds and subsequently we assume $h_{\rm inv}{\color{black}(Q)}<\infty$. We
	fix $\varepsilon>0$ and pick an invariant cover $(\mathcal{A},G)$ of $\Sigma$
	and $Q$ so that $h(\mathcal{A},G)\le h_{\rm inv}{\color{black}(Q)}+\varepsilon$.  We fix
	$\tau\in \N$ and pick a $(\tau,Q)$-spanning set $\mathcal{S}_{\rm inv}$ in $(\mathcal{A},G)$ so that 
	$N(\mathcal{S}_{\rm inv})=r_{\rm inv}(\tau,Q)$. 
	We define for every $\alpha\in\mathcal{S}_{\rm inv}$ the input sequence
	$\nu_\alpha:\intco{0;\tau}\to U$ by $\nu_\alpha(t):=G(\alpha(t))$ and introduce the
	set $\mathcal{S}_{\rm det}:=\{\nu_\alpha\mid \alpha\in\mathcal{S}_{\rm
		inv}\}$.
	For $x\in Q$ we iteratively construct $\alpha\in
	\mathcal{A}^{\intco{0;\tau}}$ and
	$\nu\in U^{\intco{0;\tau}}$ as follows:
	for $t=0$ we pick $\alpha_0\in\mathcal{S}_{\rm inv}$ so that $x\in \alpha_0(0)$
	and set $\nu(0):=G(\alpha_0(0))$. For $t\in\intco{0;\tau-1}$ we
	pick $\alpha_{t+1}\in\mathcal{S}_{\rm inv}$ so that
	$\alpha_{t+1}|_{\intcc{0;t}}=\alpha_t$ and 
	$f_{\nu|_{\intcc{0;t}}}(x)\in \alpha_{t+1}(t+1)$ and set
	$\nu(t+1):=G(\alpha_{t+1}(t+1))$. Since $(\mathcal{A},G)$ is an
	invariant cover of $(X,U,F)$ and $Q$, it is easy to show that
	$f_{\nu|_{\intcc{0;t}}}(x)\in Q$ holds for all $t\in\intco{0;\tau}$, which
	implies that
	$\mathcal{S}_{\rm det}$ is $(\tau,Q)$-spanning for $f$ and $Q$.
	Thus, we obtain $r_{\rm det}(\tau,Q)\le \no\mathcal{S}_{\rm det} \le 
	\no\mathcal{S}_{\rm inv}\le N(\mathcal{S}_{\rm inv})=r_{\rm inv}(\tau,Q)$,
	where the inequality $\no\mathcal{S}_{\rm inv}\le N(\mathcal{S}_{\rm inv})$
	follows from Lemma~\ref{l:nospanningset}.
	Since this holds for any $\tau\in\N$, we obtain the inequality  $\varepsilon+h_{\rm inv}{\color{black}(Q)}\ge
	h(\mathcal{A},G)\ge h_{\rm det}{\color{black}(Q)}$ for arbitrary $\varepsilon>0$ which shows $h_{\rm inv}{\color{black}(Q)}\ge h_{\rm det}{\color{black}(Q)}$.
	\hfill{}
\end{proof}

\subsection{Invariant covers with closed elements}

We conclude this section with a result on the topological structure of the cover
elements for systems with {\color{black}topological state alphabet and} lower semicontinuous transition functions
and closed sets $Q$. The result is used
in Theorem~\ref{t:lb} but might be of interest on its own.

Let $A$ and $B$ be topological
spaces and $f:A\rightrightarrows B$.
We say that $f$ is \emph{lower semicontinuous} if $f^{-1}(V)$
is open whenever $V\subseteq B$ is open.

{\color{black}
	\begin{theorem}\label{t:closed}
		Consider a system $\Sigma=(X,U,F)$ with topological state alphabet and a nonempty closed set 
		$Q\subseteq X$.
		Assume that $F(\cdot,u)$ is lower semicontinuous for every $u \in U$.
		%Suppose that $F(\cdot,u)$ is lower semicontinuous for every
		%$u\in U$ and $Q$ is closed.
		Let $(\mathcal{A},G)$ be an invariant cover of $\Sigma$ and $Q$ and let
		$	\mathcal{C}  :=\{ \cl A\subseteq \cl X \mid A\in \mathcal{A}\}	$.
		Then there exists a map $H^*:\mathcal{C} \rightarrow U $ such that
		$(\mathcal{C},H^*)$ is an invariant cover of $\Sigma$ and $Q$ and 												
		\begin{equation}\label{e:hch_to_hag}
		h(\mathcal{C},H^*)\le h(\mathcal{A},G).
		\end{equation}
\end{theorem}}
In the proof of the theorem, we use the following lemma.
\begin{lemma}\label{l:lsc}
	Let $X$ be a topological space and $f:X\rightrightarrows X$. If $f$ is lower
	semicontinuous then $f(\cl \Omega)\subseteq \cl f(\Omega)$ holds for every
	nonempty subset $\Omega\subseteq X$.
\end{lemma}
\begin{proof}
	For the sake of contradiction, suppose there exists $x\in \cl \Omega$, $y\in
	f(x)$ and $y\not\in \cl f(\Omega)$. 
	{\color{black}Then the open set $V:=X \setminus \cl f(\Omega) $ contains $y$. 
		%Let $M$ be any subset of $X$, then 
		Let us define 
		$U:=f^{-1}(V)= \{x'\in X\mid f(x')\cap V\neq \emptyset\}$ and since 
		$f$ is lower semicontinuous and $V$ is open so $U$ is open. 
		As $V\cap f(x) \ni y $, thus nonempty, so $x\in U$. 
		By definition, $V\cap f(\Omega)=\emptyset$ so $U\cap\Omega=\emptyset$ 
		and since $U$ is open so $U\cap \cl \Omega=\emptyset$ which is in
		contradiction with $x\in U$ and $x\in \cl\Omega$.                
		%Then there exists an open set $V$ so that
		%$y\in V$ and $V\cap
		%\cl f(\Omega)=\emptyset$. Since $f$ is lower semicontinuous it follows that
		%$U:=f^{-1}(V)$ is open and from $V\cap f(\Omega)=\emptyset$ it follows that $U$ is
		%disjoint from $\Omega$. Hence, we reach a contradiction since $x\in U\cap \cl
		%\Omega$.
	}
\end{proof}

{\color{black}\begin{proof}[Proof of Theorem~\ref{t:closed}]
		In Lemma~\ref{lem:MPG} in the Appendix, let $M=\cl$, $\Sigma_1=\Sigma_2 = \Sigma$, 
		$Q_2 = Q_1 =Q$, $\mathcal{A}_2=\mathcal{A}$, $G_2 = G$,
		$\mathcal{A}_1 = \mathcal{C}$ and $r = \id$, then one can easily verify 
		that conditions \ref{lem:MPG:cover}) - \ref{lem:MPG:subset}) hold,
		while Lemma~\ref{l:lsc} implies that~\ref{lem:MPG:semconj}) is satisfied.
		Thus there exists a map
		$H^*:\mathcal{C}\to U$ such that $(\mathcal{C},H^*)$ is an 
		invariant cover of $\Sigma$ and $Q$, and
		$h(\mathcal{C}, H^*) \leq h(\mathcal{A},G)$.
		%				-----------------\\
		%				For $C \in \mathcal{C}$,  let 
		%				\begin{equation*}
		%				G_1(C) := \{ G(A) \mid \cl A = C, A \in \mathcal{A} \}.
		%				\end{equation*}
		%				Consider for each $C\in \mathcal{C}$ the set 
		%				\begin{equation*}
		%				\mathcal{V}(C):=\left\{ (V',u) \mid V' \subseteq \mathcal{C}, u \in G_1(C), F(C,u) \subseteq \cup_{C' \in V'}C' \right\}
		%				\end{equation*}
		%				For every $C \in \mathcal{C}$, $\mathcal{V}(C) \neq \emptyset$ owing to 
		%				$(\mathcal{A},G)$ being an invariant cover and 
		%				$F(\cl A, u) \subseteq \cl F(A, u)$ which follows from the lower semicontinuity of $F(\cdot,u)$ (see Lemma~\ref{l:lsc}).
		%
		%				Let $\tau \in \N$ and $\mathcal{S}_2$ be a 
		%				$(\tau,Q)$-spanning set in $(\mathcal{A},G)$ such that 
		%				$N(\mathcal{S}_2) = r_{\text{inv}}(\tau, \mathcal{A}, G, Q)$.
		%				For $\alpha\in\mathcal{S}_2$, $t \in\intco{0;\tau-1}$ if
		%				$V' = \{\cl A \mid A \in P({\alpha}|_{\intcc{0;t}} ) \}$ 
		%				then from Lemma~\ref{l:lsc} we get $F(\cl \alpha(t), G({\alpha}(t) )) \subseteq \cup_{C \in V' } C$.
		%
		%				Using Lemma~\ref{lem:MPG} with $\Sigma_1=\Sigma_2 = \Sigma$, 
		%				%		$\Sigma_1=(\cl X, U, F)$,
		%				$Q_2 = Q_1 =Q$, $\mathcal{A}_2=\mathcal{A}$, $G_2 = G$,
		%				$\mathcal{A}_1 = \mathcal{C}$, $r = Identity$, $M(A) = \cl A$ 
		%				for $A \subseteq Q_2$, we have that there exists a map
		%				$H^*:\mathcal{C}\to U$ such that $(\mathcal{C},H^*)$ is an 
		%				invariant cover of $\Sigma$ and $Q$, and
		%				$h(\mathcal{C}, H^*) \leq h(\mathcal{A},G)$.
\end{proof} }

\section{Data-Rate-Limited Feedback}
\label{s:data-rate}

We present the data rate theorem associated with the invariance feedback entropy
of uncertain control systems. It shows that the invariance feedback entropy
is a tight lower bound of the data rate of any coder-controller scheme that
renders the set of interest invariant.  

We introduce a history-dependent definition of data rates of coder-controllers
with which we extend previously used time-invariant
\cite{NairFagniniZampieriEvans07} and time-varying
\cite{NairEvansMarrelsMoran04,ColoniusKawan09} notions.  We interpret the
history-dependent definition of data rate as a nonstochastic variant of the
notion of data rate used e.g.~in~\cite[Def.~4.1]{SilvaDerpichOstergaard11} for
noisy linear systems, defined as the average of the expected length of the
transmitted symbols in the closed loop. We motivate the particular notion of
data rate by two examples; one which illustrates that the time-varying
definition~\cite{NairEvansMarrelsMoran04} results in too large data rates and
one which shows that the notion of data rate based on the framework of
nonstochastic information theory, used in~\cite{Nair12,Nair13} for
estimation~\cite{Nair13} and control~\cite{Nair12} of linear systems, leads to
too small data rates.

\subsection{The coder-controller}

We assume that a coder for the system \eqref{e:sys} is located at the sensor side (see Fig.~\ref{f:1}),
which at every time step, encodes the current state of the system using the
finite \emph{coding alphabet} $S$. It transmits a symbol $s_t\in S$ via the discrete
noiseless channel to the controller.
The transmitted symbol $s_t\in S$ might depend on all past states and is
determined by the \emph{coder function}
\begin{equation*}
\textstyle
\gamma:\bigcup_{t\in\Z_{\ge0}}X^{\intcc{0;t}}\to S.
\end{equation*}
At time $t\in\Z_{\ge0}$, the controller received $t+1$ symbols $s_0\ldots s_t$, which
are used to determine the control input given by the
\emph{controller function} 
\begin{equation*}
\textstyle
\delta:\bigcup_{t\in\Z_{\ge0}}S^{\intcc{0;t}}\to U.
\end{equation*}
A \emph{coder-controller} for~\eqref{e:sys} is a triple $H:=(S,\gamma,\delta)$,
where $S$ is a coding alphabet and $\gamma$ and $\delta$ are a compatible coder
function and controller function, respectively. 

Given a coder-controller $(S,\gamma,\delta)$ for~\eqref{e:sys} and $\xi\in X^{\intcc{0;t}}$ with
$t\in\Z_{\ge0}$, let us use the mapping 
\begin{equation*}
\Gamma_t:X^{\intcc{0;t}}\to S^{\intcc{0;t}}
\end{equation*}
to denote the sequence $\zeta=\Gamma_t(\xi)$ of coder symbols 
generated by $\xi$, i.e., $\zeta(t')=\gamma(\xi|_{\intcc{0;t'}})$ holds for all $t'\in\intcc{0;t}$.
Subsequently, for $\zeta\in S^{\intco{0;t}}$ with $t\in\N$, we use
{\color{black}
	\begin{equation}\label{e:cc:post}
	Z(\zeta):= \{s\in S\mid \exists_{(\xi,\nu)\in \mathcal{B}(\Sigma)}\;
	\zeta s=\Gamma(\xi|_{\intcc{0;t}}) 
	\wedge \forall_{t' \in \intco{0;t}}  \nu(t')=\delta(\zeta|_{\intco{0;t'}}) \}
	\end{equation}
}to denote the possible successor coder symbols $s$ of the symbol sequence
$\zeta$ in the closed loop illustrated in Fig.~\ref{f:1}.
For notational convenience,  let us use the
convention $Z(\emptyset):=S$, so that $Z(\zeta|_{\intco{0;0}})=S$ for any
sequence $\zeta$ in $S$.
For $\tau\in \N\cup\{\infty\}$, we introduce the set 
\begin{equation*}
\mathcal{Z}_\tau:=\{ \zeta\in S^{\intco{0;\tau}}
\mid
\zeta(0)\in \gamma(X)\wedge
\forall_{t\in
	\intoo{0;\tau}}\; \zeta(t)\in
Z(\zeta|_{\intco{0;t}})\}
\end{equation*}
and define the \emph{transmission data rate} of a
coder-controller $H$ by
\begin{equation}\label{e:datarate}
R(H):=\limsup_{\tau\to\infty}
\max_{\zeta\in\mathcal{Z}_\tau }\frac{1}{\tau}
\sum_{t=0}^{\tau-1}
\log_2\no
Z(\zeta|_{\intco{0;t}})
\end{equation}
as the asymptotic average numbers of symbols in~$Z(\zeta)$ considering the worst-case
of possible symbol sequences  $\zeta\in \mathcal{Z}_\tau$.

A coder-controller
$H=(S,\gamma,\delta)$ for~\eqref{e:sys} is called \emph{$Q$-admis\-sib\-le} where $Q$ is a
nonempty subset of $X$, if for every
trajectory $(\xi,\nu)$ on $\intco{0;\infty}$ of~\eqref{e:sys} that satisfies
\begin{equation}\label{e:cc:admissible}
\xi(0)\in Q \text{ and }
\forall_{t\in\Z_{\ge0}}\;\nu(t)=\delta(\Gamma_t(\xi|_{\intcc{0;t}}))
\end{equation}
we have $\xi(\Z_{\ge0})\subseteq Q$. Let us use $\mathcal{B}_Q(H)$ to denote the
set of all trajectories $(\xi,\nu)$ on $\intco{0;\infty}$ of~\eqref{e:sys} that satisfy~\eqref{e:cc:admissible}.

\subsubsection{Time-varying data rate definition}

We follow~\cite{NairEvansMarrelsMoran04} and introduce a time-varying notion of
data rate for a coder-controller $H=(S,\gamma,\delta)$ for~\eqref{e:sys}.
Let $(S_t)_{t\ge0}$ be the sequence in 
{\color{black} the power set of}
$S$ that for each $t\in Z_{\ge0}$
contains the smallest number of symbols so that  $\gamma(\xi)\in S_t$ holds for
all $\xi\in X^{\intcc{0;t}}$. Then the time-varying data rate  of $H$ follows by 
\begin{equation*}
R_{\rm tv}(H):= \liminf_{\tau\to\infty}\frac{1}{\tau}\sum_{t=0}^{\tau-1}\log_2\no S_t.
\end{equation*}
In the following we use an example to show that there exists a $Q$-admissible
coder-controller $H$, which satisfies $R(H)<R_{\rm tv}(\bar H)$ for any 
$Q$-admissible coder-control\-ler $\bar H$.
Note that this inequality is purely a nondeterministic pheno\-me\-non: if the control system is
deterministic, it follows from the deterministic and the nondeterministic data
rate theorem (\cite[Thm.~1]{NairEvansMarrelsMoran04} and
Theorem~\ref{t:datarate} below) and  the
equivalence $h_{\rm det}{\color{black}(Q)}=h_{\rm inv}{\color{black}(Q)}$ (Theorem~\ref{t:invdet}) that the 
different notions of data rates coincide in the sense that $\inf_{H}
R(H)=\inf_{H} R_{\rm tv}(H)$ (at least if the strong invariance condition 
in \cite[Thm.~1]{NairEvansMarrelsMoran04} holds).

%i.e., the set $Q$ in the
%example can be rendered invariant with a strictly lower data rates then 
\begin{example}[label=ex:timevaryingdatarate]
	\normalfont
	Consider an instance of~\eqref{e:sys} with $U:=\{a,b\}$,  $X:=\{0,1,2,3\}$ and
	$F$ is illustrated by
	
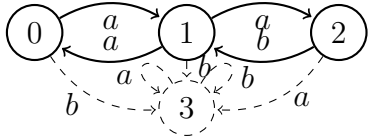
\begin{center} 
		\begin{tikzpicture}[->,thick,shorten >=1pt]
		\tikzstyle{state} = [draw, circle, minimum height=1.5em, minimum width=1.5em]
		
		\node[state] (0) at (-2,0) {$0$};
		\node[state] (1) at (0,0) {$1$};
		\node[state] (2) at (2,0) {$2$};
		\node[thin,dashed,state] (3) at (0,-1) {$3$};
		
		\path (0) edge [bend left] node[below] {$a$} (1);
		\path (1) edge [bend left] node[above] {$a$} (0);
		
		\path (2) edge [bend left] node[above] {$b$} (1);
		\path (1) edge [bend left] node[below] {$a$} (2);

		\path[thin,dashed] (3) edge [in=30,out=60,loop] node[right,near end] {$b$} (3);
		\path[thin,dashed] (3) edge [in=150,out=120,loop] node[left,near end] {$a$} (3);
		%\path (3) edge [loop left ] node {$a$} (3);
		
		\path[thin,dashed] (0) edge[bend right]  node[near start,below] {$b$} (3);
		\path[thin,dashed] (2) edge[bend left ]  node[near start,below] {$a$} (3);
		\path[thin,dashed] (1) edge  node[near start,right] {$b$} (3);
		\end{tikzpicture}
	\end{center}
	Let $Q:=\{0,1,2\}$.  The transitions that lead outside $Q$ and the
	states that are outside $Q$ are marked by dashed lines.
	Consider the coder-cont\-roller
	$H=(S,\gamma,\delta)$  with
	$S:=X$ and $\gamma$ and $\delta$ are given
	for $\xi\in X^{\intcc{0;t}}$, $t\in\Z_{\ge0}$, by $\gamma(\xi):=\xi(t)$
	and $\delta(\xi):=a$ if $\xi(t)\in\{0,1,3\}$  and $\delta(\xi):=b$ if
	$\xi(t)=2$. We compute the number of possible successor symbols $Z(\xi)$ for $\xi\in
	X^{\intcc{0;t}}$, $t\in\Z_{\ge0}$, by $\no Z(\xi)=1$ if $\xi(t)\in\{0,2,3\}$
	and  $\no Z(\xi)=2$ if $\xi(t)=1$. It is easy to verify that $H$ is
	$Q$-admissible. Since the state $\xi(t)=1$ occurs only
	every other time step for any element $(\xi,\nu)$ of the closed loop, we compute the data rate to
	$R(H)=\nicefrac{1}{2}$. 
	Consider a time-varying $Q$-admissible coder-controller $\bar H=(\bar
	S,\bar\gamma,\bar\delta)$. Initially, the states $\{0,1\}$ and $\{2\}$ need to
	be distinguishable at the
	controller side in order to confine the system to $Q$ so that $\no \bar S_0\ge 2$
	follows. At time $t=1$,
	the system is possibly again in any of the states $\{0,1,2\}$ (depending on the
	initial condition) and we have $\no \bar S_1\ge 2$. By continuing this
	argument we see that $\no\bar  S_t\ge 2$ for all $t\in\Z_{\ge0}$ and 
	$R_{\rm tv}(\bar H)\ge 1$ follows.
	\qed
\end{example}

\subsubsection{Zero-error capacity of uncertain channels}

Alternatively to the definition of the data rate of a coder-controller
in~\eqref{e:datarate} we could follow~\cite{Nair12,Nair13} and define the data
rate of a coder-controller as the zero-error capacity $C_0$ of an \emph{ideal
	stationary memoryless uncertain channel} (SMUC) in the nonstochastic information theory
framework presented in~\cite[Def.~4.1]{Nair13}. The input alphabet of the SMUC equals
the output alphabet and is given by $S$. The channel is ideal and does not
introduce any error in the transmission. Hence, the transition function is the identity, i.e., $T(s)=s$ holds for all $s\in S$. The input function space
$\mathcal{Z}_\infty\subseteq S^{\intco{0;\infty}}$ is the set of all possible
symbol sequences that are generated by the closed loop, which represents the
total amount of information that needs to be
transmitted by the channel. For the ideal SMUC, the zero-error
capacity~\cite[Eq.~(25)]{Nair13}, for a
coder-controller $H$ results in
\begin{equation*}\label{eq:ZeroErrorCapacity}
C_0(H):=\lim_{\tau\to\infty} \frac{1}{\tau} \log_2\no\mathcal{Z}_\tau.
\end{equation*}
We use the following example to demonstrate that the zero-error capacity is too
low, i.e., $C_0(H)=0$ while $R(H)\ge 1$.

\begin{example}[label=ex:zeroerrorcap]
	\normalfont
	Consider an instance of~\eqref{e:sys} with $U:=\{a,b\}$,  $X:=\{0,1,2,3\}$ and
	$F$ is illustrated by
	
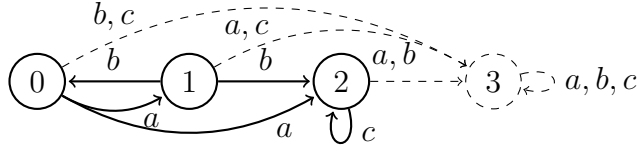
\begin{center} 
		\begin{tikzpicture}[->,thick,shorten >=1pt]
		\tikzstyle{state} = [draw, circle, minimum height=1.5em, minimum width=1.5em]
		
		\node[state] (0) at (-2,0) {$0$};
		\node[state] (1) at (0,0) {$1$};
		\node[state] (2) at (2,0) {$2$};
		\node[thin,dashed,state] (3) at (4,0) {$3$};
		
		\path[thin,dashed] (0) edge [bend left] node[very near start,above] {$b,c$} (3);
		\path[thin,dashed] (1) edge [bend left] node[very near start,above] {$a,c$} (3);
		\path[thin,dashed] (2) edge             node[near start,above] {$a,b$} (3);
		
		\path (2) edge [loop below] node[near start,right] {$c$} (2);
		
		\path (0) edge[bend right] node[very near end,below] {$a$} (1);
		\path (0) edge[bend right] node[very near end,below] {$a$} (2);
		
		\path (1) edge node[above] {$b$} (0);
		\path (1) edge node[above] {$b$} (2);
		
		\path[thin,dashed] (3) edge [loop right] node {$a,b,c$} (3);
		\end{tikzpicture}
	\end{center}
	The transitions and states that lead, respectively, are outside the set of
	interest $Q:=\{0,1,2\}$ are dashed.
	Consider the $Q$-admissible coder-cont\-roller
	$H=(S,\gamma,\delta)$  with
	$S:=X$ and $\gamma$ and $\delta$ are given
	for $\xi\in X^{\intcc{0;t}}$, $t\in\Z_{\ge0}$ by $\gamma(\xi):=\xi(t)$
	and 
	\begin{equation*}
	\delta(\xi):=
	\begin{cases}
	a & \text{if } \xi(t)\in\{0,3\}\\
	b & \text{if } \xi(t)=1\\
	c & \text{if } \xi(t)=2.
	\end{cases}
	\end{equation*}
	We pick the trajectory $(\xi,\nu)\in \mathcal{B}_{Q}(H)$ given for
	$t\in\Z_{\ge0}$ by $\xi(2t)=0$ and $\xi(2t+1)=1$. We obtain
	$Z(\xi|_{\intcc{0;t}})=\{1,2\}$ if $\xi(t)=0$ and 
	$Z(\xi|_{\intcc{0;t}})=\{0,2\}$ if $\xi(t)=1$. Since $\no F(x,u)\le 2$ for all
	$x\in X$ and $u\in U$, it is straightforward to see that $\sum_{t=0}^{\tau-1} \log_2\no
	Z(\xi|_{\intco{0;t}})=
	\max_{\zeta\in\mathcal{Z}_\tau }
	\sum_{t=0}^{\tau-1}
	\log_2\no
	Z(\zeta|_{\intco{0;t}})$ holds for all $\tau\in\N$. Hence, we obtain
	$R(H)=1$.\\
	We are going to derive $C_0(H)$. Consider the set $\mathcal{Z}_\tau\subseteq X^{\intco{0;\tau}}$
	and the hypothesis for $\tau\in\N$: there exists at most one $\xi\in\mathcal{Z}_\tau$ with
	$\xi(\tau-1)=1$ and there exists at most one $\xi\in\mathcal{Z}_\tau$
	with $\xi(\tau-1)=0$. For $\tau=1$ we have $\mathcal{Z}_1=X$ and the
	hypothesis holds. Suppose the hypothesis holds for $\tau\in \N$ and let
	$\xi\in\mathcal{Z}_\tau$. We have $Z(\xi)=\{0,2\}$ if $\xi(t)=1$, 
	$Z(\xi)=\{1,2\}$ if $\xi(t)=0$, 
	$Z(\xi)=\{2\}$ if $\xi(t)=2$ and
	$Z(\xi)=\{3\}$ if $\xi(t)=3$, so that the hypothesis holds for $\tau+1$, which
	shows that the hypothesis holds for every $\tau\in\N$. Therefore, we obtain a
	bound of the number of elements in $\mathcal{Z}_\tau$ by $4+2(\tau-1)$ and
	the zero-error capacity of $H$ follows by $C_0(H)=0$.
	\qed
\end{example}

Example~\ref{ex:zeroerrorcap} shows that even though, the asymptotic
average of the total amount of information that needs to be transmitted (= symbol
sequences generated by the closed loop) via the
channel is zero, the necessary (and sufficient) data rate to confine the system
$\Sigma$ within $Q$ is one. The discrepancy results from the causality
constraints that are imposed on the coder-controller structure by the invariance
condition, i.e., at each instant in time the controller needs to be able to
produce a control input so that all successor states are inside $Q$ see
e.g.~\cite{SilvaDerpichOstergaard11}. Contrary to this observation, the zero-error capacity is an
adequate measure for data rate constraints for deterministic linear systems (without
disturbances)~\cite{Nair12,Nair13}.
%as well as for nonlinear, deterministic
%control systems, given that we interpret the  $(\tau,Q)$-spanning set
%$\mathcal{S}$ with minimal cardinality as the total amount of information that needs to be transmitted until time $\tau$, see~\eqref{e:dt:min:span}.

\subsubsection{Periodic coder-controllers}
In the proof of the data rate theorem, we work with periodic coder-controllers.
Given $\tau\in\N$ and a coder-controller $H=(S,\gamma,\delta)$, we say that $H$
is \emph{$\tau$-periodic} if for all $t\in\Z_{\ge0}$, 
$\zeta\in S^{\intcc{0;t}}$ and 
$\xi\in X^{\intcc{0;t}}$ we have 
\begin{equation}\label{e:tauperiodic}
\begin{split}
\gamma(\xi)&=\gamma(\xi|_{\intcc{\tau\lfloor t/\tau \rfloor;t}}),\\
\delta(\zeta)&=\delta(\zeta|_{\intcc{\tau\lfloor t/\tau \rfloor;t}}). 
\end{split}
\end{equation}
\begin{lemma}\label{l:tauperiodic:rate}
	The transmission data rate of a $\tau$-periodic co\-der-controller 
	$H=(S,\gamma, \delta)$ for~\eqref{e:sys} is given by
	\begin{equation}\label{e:cc:tauperiodic}
	R( H)= \max_{\zeta\in\mathcal{Z}_\tau}\frac{1}{\tau}\sum_{t=0}^{\tau-1}\log_2\no Z(\zeta|_{\intcc{0;t}}).
	\end{equation}
\end{lemma}
\begin{proof}
	Let $L$ denote the right-hand-side of~\eqref{e:cc:tauperiodic}.
	Consider $T\in\N$, $\zeta\in \mathcal{Z}_T$ and set $a:=\lfloor T/\tau \rfloor$
	and $\bar \tau :=T-\tau a$.
	We define $\zeta_i:=\zeta|_{\intco{i\tau;(i+1)\tau}}$ for $i\in\intco{0;a}$
	and $\zeta_a:=\zeta|_{\intco{a\tau;T}}$. 
	Since $\gamma$ is
	$\tau$-periodic, we see that each $\zeta_i$ with $i\in\intco{0;a}$ is an
	element of $\mathcal{Z}_\tau$, and we obtain for $N_i:=\sum_{t=0}^{\tau-1} \log_2\no
	Z(\zeta_i|_{\intcc{0;t}})$ the bound $N_i\le L\tau$ for all
	$i\in\intco{0;a}$. We define $N_a:=\sum_{t=0}^{\bar\tau-1} \log_2\no
	Z(\zeta_a|_{\intcc{0;t}})$ which is bounded by $N_a\le \tau\log_2\no S$.
	Note that $a\tau+\bar \tau=T$, so that for $C:=\tau\log_2\no
	S$ we have 
	\begin{equation*}
	\textstyle
	\frac{1}{T}\sum_{t=0}^{T-1} \log_2\no Z(\zeta|_{\intcc{0;t}})
	=
	\textstyle
	\frac{1}{T}(\sum_{i=0}^{a-1}N_i+N_a) 
	\le 
	\textstyle
	\frac{1}{T}(aL\tau + L\bar\tau +C)=L+\frac{C}{T}.
	\end{equation*}
	Since $C$ is independent of $T$, the assertion follows.
\end{proof}

\begin{lemma}\label{l:tauperiodic:exists}
	For every coder-controller $H=(S,\delta,\gamma)$ for~\eqref{e:sys} and
	$\varepsilon>0$, there
	exists a $\tau$-periodic coder-controller $\hat H=(S,\hat\delta,\hat\gamma)$
	that satisfies
	\begin{equation*}
	R(\hat H) \le R(H) +\varepsilon.
	\end{equation*}
\end{lemma}
\begin{proof}
	For $\varepsilon>0$, we pick $\tau\in\N$ so that
	$\log_2\no\mathcal{Z}_0/\tau\le\nicefrac{\varepsilon}{2}$ and 
	\begin{equation*}
	\max_{\zeta\in\mathcal{Z}_\tau}
	\frac{1}{\tau}\sum_{t=0}^{\tau-2}\log_2\no Z(\zeta|_{\intcc{0;t}})\le R(H)+\nicefrac{\varepsilon}{2}. 
	\end{equation*}
	We define
	$\hat \gamma$ and $\hat\delta$ for all $\xi\in X^{\intcc{0;t}}$, $\zeta\in
	S^{\intcc{0;t}}$ with $t\in\Z_{\ge0}$ by 
	\begin{equation*}
	\hat\gamma(\xi):=\gamma(\xi|_{\intcc{\tau\lfloor t/\tau \rfloor;t}})
	\;\text{ and }\;
	\hat\delta(\zeta):=\delta(\zeta|_{\intcc{\tau\lfloor t/\tau \rfloor;t}}). 
	\end{equation*}
	Let $\hat Z$ be defined in \eqref{e:cc:post} w.r.t.~$\hat\gamma$. Then we have
	for all $\zeta\in S^{\intcc{0;t}}$ with $t\in\intco{0;\tau-1}$ the equality
	$Z(\zeta)=\hat Z(\zeta)$ and for every $\zeta\in S^{\intco{0;\tau}}$ we have $\hat
	Z(\zeta)=\mathcal{Z}_0$ which follows from the fact that $\hat\gamma$ is
	$\tau$-periodic. The transmission data rate of $\hat H$ follows
	by~\eqref{e:cc:tauperiodic} which is bounded by
	\begin{equation*}
	\max_{\zeta\in\hat{\mathcal{Z}}_{\tau}}
	\frac{1}{\tau}(\sum_{t=0}^{\tau-2}\log_2\no \hat
	Z(\zeta|_{\intcc{0;t}})+\log_2\no \mathcal{Z}_0)\le
	R(H)+\varepsilon.\qedhere
	\end{equation*}
\end{proof}

\subsection{The data rate theorem}
The next result establishes the main data rate theorem in this work.

\begin{theorem}\label{t:datarate}
	Consider the system $\Sigma=(X,U,F)$ and a nonempty set
	$Q\subseteq X$. The invariance feedback entropy of $\Sigma$ and $Q$ satisfies
	\begin{equation}\label{e:datarate:equiv}
	h_{\rm inv}{\color{black}(Q)} = \inf_{H\in \mathcal{H}} R(H)
	\end{equation}
	where $\mathcal{H}$ is the set of all $Q$-admissible coder-controllers
	for~$\Sigma$.
\end{theorem}

We use the following two technical lemmas to show the theorem.

\begin{lemma}\label{l:datarate:lowerbound}
	Let $H=(S,\gamma,\delta)$ be a $Q$-admissible $\tau$-periodic coder-controller
	for~$\Sigma=(X,U,F)$. Then there exists an invariant cover $(\mathcal{A},G)$ of
	$\Sigma$ and $Q$ and a $(\tau,Q)$-spanning set $\mathcal{S}$ in $(\mathcal{A},G)$ so that
	%the associated expansion number equals the transmission rate of $H$, i.e., 
	\begin{equation*}
	\frac{1}{\tau}\log_2 N(\mathcal{S})\le  R(H).
	\end{equation*}
\end{lemma}
\begin{proof}
	For every $t\in\intco{0;\tau}$ and every $\zeta\in \mathcal{Z}_{t+1}$ we define
	$A(\zeta)
	:=
	\{x\in Q
	\mid 
	\exists_{(\xi,\nu)\in \mathcal{B}_Q(H)}\;
	\zeta=\Gamma_t(\xi|_{\intcc{0;t}})
	\wedge
	\xi(t)=x\}$,
	$G(A(\zeta)):=\delta(\zeta)$ and
	$\mathcal{A}:=\{A(\zeta)\mid \zeta\in \mathcal{Z}_{t+1}\wedge
	t\in\intco{0;\tau} \}$. We show that $(\mathcal{A},G)$ is an
	invariant cover of $\Sigma$ and $Q$. Clearly, $\mathcal{A}$ is finite and every element of $\mathcal{A}$ is a
	subset of $Q$.  Since $H$ is $Q$-admissible, for every $x\in Q$ there
	exists $(\xi,\nu)\in\mathcal{B}_Q(H)$ so that $\xi(0)=x$. Hence, $\{ A(s)\mid
	s\in \mathcal{Z}_1\}$ covers $Q$ and we see that $\mathcal{A}$ covers $Q$. Let $A\in \mathcal{A}$ and
	suppose that there exists $x\in A$ so that $F(x,G(A))\not\subseteq Q$. 
	Since $A\in\mathcal{A}$, there exists $t\in\intco{0;\tau}$,
	$\zeta\in \mathcal{Z}_{t+1}$ and $(\xi,\nu)\in\mathcal{B}_Q(H)$ so that
	$A=A(\zeta)$, $\zeta=\Gamma_t(\xi|_{\intcc{0;t}})$ and $x=\xi(t)$. Note that $\nu$
	satisfies \eqref{e:cc:admissible} so that $\nu(t)=G(A(\zeta))$ holds.
	We fix $x'\in F(x,G(A))\smallsetminus Q$ and pick a trajectory
	$(\xi',\nu')$ of~$\Sigma$ on $\intco{0;\infty}$ such that $\xi'(0)=x'$ and 
	$\nu'(t')=\delta(\Gamma_t((\xi|_{\intcc{0;t}}\xi')|_{\intcc{t;t+t'+1}}))$ holds
	for all $t'\in\Z_{\ge0}$. 
	We define $(\bar\xi,\bar\nu)$ by $\bar\xi:=\xi|_{\intcc{0;t}}\xi'$
	and $\bar\nu:=\nu|_{\intcc{0;t}}\nu'$, which by construction is a trajectory
	of~$\Sigma$ on $\intco{0;\infty}$
	which satisfies \eqref{e:cc:admissible} but
	$\bar\xi(\intco{0;\infty})\not\subseteq Q$. This contradicts the
	$Q$-admissibility of $H$ and we can deduce that $F(A,G(A))\subseteq Q$ for all
	$A\in\mathcal{A}$, which shows that $(\mathcal{A},G)$ is an invariant cover of $\Sigma$ and
	$Q$. \\
	% 
	%Let us show that $\no\mathcal{Z}_\tau\le \max_{\zeta\in\mathcal{Z}_\tau}\prod_{t=0}^{\tau-1}\no Z(\zeta|_{\intcc{0;t}})$,
	%\\[2cm]
	%
	%  Let us use $\mathcal{Z}$ to denote the set of all sequences $\zeta\in
	%  S^{\intco{0;\tau}}$ that satisfy 
	%  $\zeta(t+1)\in Z(\zeta|_{\intcc{0;t}})$ 
	% % $F(A(\zeta|_{\intcc{0;t}}),G(A(\zeta|_{\intcc{0;t}}))\cap A(\zeta|_{\intcc{0;t+1}})\neq \emptyset$
	%  for all $t\in\intco{0;\tau-1}$. For each $s\in S$, the number of sequences in
	%  $\zeta\in\mathcal{Z}$ with $\zeta(0)=s$ is bounded by $\max_{\zeta\in
	%  S^{\intco{0;\tau}}}\prod_{t=1}^{\tau-1} \no Z(\zeta|_{\intco{0;t}})$ and we
	%  can bound $\no\mathcal{Z}$ by $\max_{\zeta\in
	%  S^{\intco{0;\tau}}}\prod_{t=0}^{\tau-1} \no Z(\zeta|_{\intco{0;t}})$.
	We are going to construct a $(\tau,Q)$-spanning set $\mathcal{S}\subseteq
	\mathcal{A}^{\intco{0;\tau}}$ with the help of $\mathcal{Z}_\tau$.
	For each
	$\zeta\in \mathcal{Z}_\tau$ we define a sequence $\alpha_\zeta:\intco{0;\tau}\to
	\mathcal{A}$ by $\alpha_\zeta(t):=A(\zeta|_{\intcc{0;t}})$ for all
	$t\in\intco{0;\tau}$ and use $\mathcal{S}$ to denote the set of all such
	sequences $\{\alpha_\zeta\mid \zeta\in\mathcal{Z}_\tau\}$.
	Note that $P(\alpha_\zeta)=\{A(s)\mid s\in\mathcal{Z}_1\}$ holds for all
	$\alpha_\zeta\in\mathcal{S}$, and we see that $P(\alpha_\zeta)$ covers $Q$.
	Let us show~\eqref{e:icover:spanning}. Let $\alpha_\zeta\in\mathcal{S}$,
	$t\in\intco{0;\tau-1}$ so that $\alpha_\zeta(t)= A(\zeta|_{\intcc{0;t}})$. We define
	$\zeta_{t}:=\zeta|_{\intcc{0;t}}$ and fix  $x_0\in A(\zeta_t)$ and $x_1\in
	F(x_0,G(A(\zeta_{t})))$. Since $x_0\in A(\zeta_t)$ there exists
	$(\xi,\nu)\in\mathcal{B}_Q(H)$ so that $\zeta_t=\Gamma_t(\xi|_{\intcc{0;t}})$
	with $\xi(t)=x_0$ and we use~\eqref{e:cc:admissible} to see that
	$G(A(\zeta_t))=\delta(\zeta_t)=\nu(t)$. Therefore, $(\xi,\nu)|_{\intcc{0;t}}$
	can be extended to a trajectory in $(\bar\xi,\bar\nu)\in\mathcal{B}_Q(H)$ with $\bar\xi(t+1)=x_1$.
	Let $s=\gamma(\bar\xi|_{\intcc{0;t+1}})$, then we have $s\in Z(\zeta_t)$ and $\zeta_{t+1}:=\zeta_t s\in\mathcal{Z}_{t+2}$ holds. Moreover, 
	$\zeta_{t+1}=\Gamma_{t+1}(\bar\xi|_{\intcc{0;t+1}})$ and we conclude 
	that $x_1\in A(\zeta_{t+1})$. We repeat this process for $x_i\in
	F(A(\zeta_{t+i}),G(A(\zeta_{t+i}))$, $i\in\intcc{0;k}$ until $t+k=\tau-1$ at which
	point we arrive at $\zeta_{t+k}\in \mathcal{Z}_\tau$ and we see that the
	associated sequence $\alpha_{\zeta_{t+k}}$ is an element of $\mathcal{S}$ that satisfies $x_1\in \alpha_{\zeta_{t+k}}(t+1)$ and $\alpha_{\zeta_{t+k}}|_{\intcc{0;t}}=\alpha_\zeta|_{\intcc{0;t}}$. Since such a
	sequence can be constructed for every $x_1\in F(x_0,G(A(\zeta_{t})))$ and $x_0\in
	A(\zeta_{t})$,
	we see that~\eqref{e:icover:spanning} holds and it follows that $\mathcal{S}$
	is $(\tau,Q)$-spanning in $(\mathcal{A},G)$.\\
	We claim that $\no
	P(\alpha_\zeta|_{\intcc{0;t}})\le \no Z(\zeta|_{\intcc{0;t}})$
	for every $\alpha_\zeta\in\mathcal{S}$
	and $t\in\intco{0;\tau-1}$. Let $A\in
	P(\alpha_\zeta|_{\intcc{0;t}})$, then there exists
	$\alpha_{\zeta'}\in\mathcal{S}$ such that
	$A=\alpha_{\zeta'}(t+1)$ and
	$\zeta'|_{\intcc{0;t}}=\zeta|_{\intcc{0;t}}$. Hence $\zeta'(t+1)\in
	Z(\zeta|_{\intcc{0;t}})$. Moreover, for $A,\bar A\in
	P(\alpha_\zeta|_{\intcc{0;t}})$ with $A\neq \bar A$ there exists
	$\alpha_{\zeta'},\alpha_{\bar\zeta'}\in\mathcal{S}$ such that
	$A=A(\zeta'|_{\intcc{0;t+1}})$ and $\bar A=A(\bar\zeta'|_{\intcc{0;t+1}})$,
	which shows that $\zeta'(t+1)\neq \bar\zeta'(t+1)$ and $\zeta'(t+1), \bar\zeta'(t+1)\in Z(\zeta|_{\intcc{0;t}})$
	and we obtain $\no
	P(\alpha_\zeta|_{\intcc{0;t}})\le \no Z(\zeta|_{\intcc{0;t}})$ for all
	$t\in\intco{0;\tau-1}$ and $\zeta\in\mathcal{Z}_\tau$. For $t=\tau-1$ we have
	$P(\zeta)=\{A(s)\mid s\in \mathcal{Z}_1\}$. For $Z(\zeta)$ we have
	$Z(\zeta)=\gamma(X)$, since $H$ is $\tau$-periodic and we obtain $\no P(\alpha_\zeta)\le \no
	Z(\zeta)$ for every $\zeta\in\mathcal{Z}_\tau$. Hence,  $N(\mathcal{S})\le \max_{\zeta\in \mathcal{Z}_\tau}\prod_{t=0}^{\tau-1}\no
	Z(\zeta|_{\intcc{0;t}})$ follows and we obtain $\frac{1}{\tau}\log_2
	N(\mathcal{S})\le R(H)$.
\end{proof}

In the proof of the following lemma, we use an \emph{enumeration} of a finite
set $A$, which is a function $e:\intcc{1;\no A}\to A$
such that $e(\intcc{1;\no A})=A$.

\begin{lemma}\label{l:datarate:upperbound}
	Consider an invariant cover $(\mathcal{A},G)$ of~$\Sigma=(X,U,F)$ and some
	nonempty set $Q\subseteq X$. Let $\mathcal{S}$ be a $(\tau,Q)$-spanning set
	in $(\mathcal{A},G)$. Then there exists a $Q$-admissible $\tau$-periodic
	coder-controller $H=(S,\gamma,\delta)$ for~$\Sigma$ so that 
	%the associated expansion number equals the transmission rate of $H$, i.e., 
	\begin{equation*}
	\frac{1}{\tau}\log_2 N(\mathcal{S})\ge R(H).
	\end{equation*}
\end{lemma}
\begin{proof}
	We define 
	$\mathcal{S}_t:=
	\{
	\alpha\in\mathcal{A}^{\intcc{0;t}}
	\mid 
	\exists_{\hat\alpha\in \mathcal{S}}\;
	\hat\alpha|_{\intcc{0;t}}=\alpha
	\}$ for $t\in\intco{0;\tau}$ and observe that $\mathcal{S}_{\tau-1}=\mathcal{S}$ and for
	every $\alpha\in\mathcal{S}$ we have $P(\alpha)=\mathcal{S}_0$.
	For $\alpha\in \mathcal{S}_t$ with $t\in\intco{0;\tau-1}$ let
	$e(\alpha)$ be an enumeration of $P(\alpha)$. We slightly abuse the notation,
	and use $e(\emptyset)$ to denote an enumeration
	of $\mathcal{S}_0$ so that  $e(\alpha|_{\intco{0;0}})=e(\emptyset)$ for all 
	$\alpha\in\mathcal{S}$.
	Let $m\in\N$ be the smallest number so that every
	co-domain of $e(\alpha)$ is a subset of $\intcc{1;m}$. We use this interval to define
	the set of symbols $S:=\intcc{1;m}$.
	We are going to define $\gamma(\xi)$ and $\delta(\zeta)$ for
	all sequences $\xi\in X^{\intcc{0;t}}$, respectively, $\zeta\in S^{\intcc{0;t}}$
	with $t\in\intco{0;\tau}$, which determines $\gamma$ and $\delta$ for all
	elements in their domain, since $\gamma$ and $\delta$ are $\tau$-periodic.
	We begin with $\gamma$, which we define iteratively. For $t=0$ and $x\in X$ we
	set $\gamma(x):=e(\emptyset)(A)$ if there exists $A\in \mathcal{S}_0$ with $x\in A$.
	If there are several $A\in \mathcal{S}_0$ that
	contain $x$ we simply pick one. If there does not exist any
	$A\in\mathcal{S}_0$ with $x\in A$ we set $\gamma(x):=1$. For
	$t\in\intoo{0;\tau}$ and $\xi\in X^{\intcc{0;t}}$ we define
	$\gamma(\xi):=e(\alpha|_{\intco{0;t}})(\alpha(t))$ for 
	$\alpha\in\mathcal{S}_t$ that satisfies
	i) $\xi(t)\in \alpha(t)$ and ii)
	$\gamma(\xi|_{\intcc{0;t'}})=e(\alpha|_{\intco{0;t'}})(\alpha(t'))$ holds for
	all $t'\in\intco{0;t}$. Again, if there are several such
	$\alpha\in\mathcal{S}_t$ we simply pick one. If there does not
	exist any $\alpha$ in $\mathcal{S}_t$ that satisfies i) and ii), we set
	$\gamma(\xi):=1$. We define $\delta$ for $t\in\intco{0;\tau}$ and
	$\zeta\in S^{\intcc{0;t}}$ as follows: if there exists
	$\alpha\in\mathcal{S}_t$ that satisfies
	$e(\alpha|_{\intco{0;t'}})(\alpha(t'))=\zeta(t')$ for all $t'\in\intcc{0;t}$,
	we set $\delta(\zeta):=G(\alpha(t))$, otherwise we set
	$\delta(\zeta):=u$ for some $u\in U$. 
	Let us show that the
	coder-controller is $Q$-admissible. We fix
	$(\xi,\nu)\in\mathcal{B}_Q(H)$ and proceed by induction with the hypothesis
	parameterized by $t\in\intco{0;\tau}$ : there exists $\alpha\in\mathcal{S}_t$ so that 
	$\xi(t)\in\alpha(t)$, $\gamma(\xi|_{\intcc{0;t'}})=e(\alpha|_{\intco{0;t'}})(\alpha(t'))$ 
	and $\nu(t')=G(\alpha(t'))$ hold for all
	$t'\in\intcc{0;t}$. For $t=0$, we know that $\mathcal{S}_0$ covers $Q$ so that
	for $\xi(0)\in Q$ there exists $A\in\mathcal{S}_0$
	with $x\in A$ and it follows from the definition of $\gamma$ and $\delta$ that $\gamma(\xi(0))=e(\emptyset)(\bar A)$ for some $\bar A\in\mathcal{S}_0$ with $\xi(0)\in \bar A$ and $\nu(0)=\delta(\gamma(\bar
	A))=G(\bar A)$. 
	Now suppose that the induction hypothesis holds for
	$t\in\intoo{0;\tau-1}$. Since $\xi(t)\in\alpha(t)$ and $\nu(t)=G(\alpha(t))$ for some
	$\alpha\in\mathcal{S}_t$, we use~\eqref{e:icover:spanning} to see that there exists $\bar\alpha\in\mathcal{S}$
	so that $\bar\alpha|_{\intcc{0;t}}=\alpha$ and $\xi(t+1)\in\bar\alpha(t+1)$,
	so that $\bar \alpha$ satisfies i) and ii) in the definition of $\gamma$ and
	we have $\gamma(\xi|_{\intcc{0;t+1}})=e(\alpha)(\hat\alpha(t+1))$ for some
	$\hat\alpha\in\mathcal{S}_{t+1}$ with $\xi(t+1)\in \hat \alpha(t+1)$ and
	$\hat\alpha|_{\intcc{0;t}}=\alpha$.
	Since $\hat\alpha$ is uniquely determined
	by the symbol sequence $\zeta\in S^{\intcc{0;t+1}}$ given by $\zeta(t')=e(\hat\alpha|_{\intco{0;t'}})(\hat\alpha(t'))$ for all
	$t'\in\intcc{0;t+1}$, we have $\nu(t+1)=\delta(\zeta)=G(\hat\alpha(t+1))$, which completes
	the induction. Note that the induction hypothesis implies that
	$F(\xi(t),\nu(t))\subseteq Q$  for all  $t\in\intco{0;\tau}$, since $\xi(t)\in\alpha(t)$ and
	$\nu(t)=G(\alpha(t))$. We obtain $\xi(\intco{0;\infty})\subseteq Q$ from the
	$\tau$-periodicity of $H$ and the $Q$-admissibility follows. \\
	We derive a bound for
	$R(H)$. Since $H$ is $\tau$-periodic, we have for any $\zeta\in
	\mathcal{Z}_\tau$ the equality $Z(\zeta)=e(\emptyset)(\mathcal{S}_0)$ and we see that $\no
	Z(\zeta)=\no e(\emptyset)(\mathcal{S}_0)=\no P(\alpha)$ for any
	$\alpha\in\mathcal{S}$. We fix $\zeta\in\mathcal{Z}_\tau$ and pick $\alpha\in\mathcal{S}$
	so that $\alpha(t)=e^{-1}(\alpha|_{\intco{0;t}})(\zeta(t))$ holds for all
	$t\in\intco{0;\tau}$. By definition, the set $Z(\zeta|_{\intcc{0;t}})$ is the
	co-domain of an enumeration of $P(\alpha|_{\intcc{0;t}})$, which shows $\no
	Z(\zeta|_{\intcc{0;t}})=\no P(\alpha|_{\intcc{0;t}})$. Therefore, we have
	$\max_{\zeta\in \mathcal{Z}_\tau}\prod_{t=0}^{\tau-1}\no Z(\zeta|_{\intcc{0;t}})\le
	\max_{\alpha\in\mathcal{S}} \prod_{t=0}^{\tau-1}\no P(\alpha|_{\intcc{0;t}})$
	and the assertion follows by~\eqref{e:cc:tauperiodic}.
\end{proof}

We continue with the proof of Theorem~\ref{t:datarate}.

\begin{proof}[Proof of Theorem~\ref{t:datarate}]
	Let us first prove the inequality  $h_{\rm inv}{\color{black}(Q)} \le \inf_{H\in \mathcal{H}} R(H)$. If the right-hand-side
	of~\eqref{e:datarate:equiv} equals infinity the inequality trivially holds and subsequently we
	assume the right-hand-side of~\eqref{e:datarate:equiv} is finite.  We fix
	$\varepsilon>0$ and pick a coder-controller $\bar H=(S,\bar\gamma,\bar \delta)$ so that
	$R(\bar H)\le \inf_{H\in\mathcal{H}} R(H)+\varepsilon$. According to
	Lemma~\ref{l:tauperiodic:exists} there exists a $\tau$-periodic coder-controller
	$H=(S,\gamma, \delta)$ so that $R(H)\le R(\bar H)+\varepsilon$.
	It is straightforward to see that for every $(\xi,\nu)\in\mathcal{B}_Q(H)$ and
	$\xi_i:=\xi|_{\intco{i\tau;(i+1)\tau}}$, $i\in\Z_{\ge0}$, there exists
	$(\bar\xi,\bar\nu)\in\mathcal{B}_Q(\bar H)$, so that
	$\xi_i=\bar\xi|_{\intco{0;\tau}}$, which shows that $H$ is $Q$-admissible.
	From Lemma~\ref{l:datarate:lowerbound} it follows that there exists an
	$(\mathcal{A},G)$ of
	$\Sigma$ and $Q$ and a $(\tau,Q)$-spanning set in $(\mathcal{A},G)$ so that
	$ \frac{1}{\tau}\log_2 N(\mathcal{S})\le R(H)$. 
	We use Lemma~\ref{l:subadditivity} to see that $r_{\rm
		inv}(n\tau,Q)\le n r_{\rm inv}(\tau,Q)$ so that
	$h(\mathcal{A},G)
	=
	\lim_{n\to \infty} \frac{1}{n\tau} \log_2 r_{\rm inv}(n\tau,Q)
	\le
	\frac{1}{\tau}\log_2 r_{\rm inv}(\tau,Q)
	\le
	\frac{1}{\tau}\log_2N(\mathcal{S})$.   By the choice of $H$ we
	obtain $2\varepsilon+\inf_{H\in\mathcal{H}} R(H)\ge R(H)\ge h_{\rm inv}{\color{black}(Q)}$.
	Since this holds for arbitrary $\varepsilon>0$ we arrive at the desired
	inequality.
	
	We continue with the inequality $h_{\rm inv}{\color{black}(Q)} \ge \inf_{H\in \mathcal{H}} R(H)$. 
	If $h_{\rm inv}{\color{black}(Q)}=\infty$ the inequality trivially holds and subsequently we
	consider $h_{\rm inv}{\color{black}(Q)}<\infty$.  We fix
	$\varepsilon>0$ and pick an invariant cover $(\mathcal{A},G)$ of $\Sigma$ and
	$Q$ so that $h(\mathcal{A},G)<h_{\rm inv}{\color{black}(Q)}+\varepsilon$. We pick $\tau\in \N$
	so that $\frac{1}{\tau}\log_2 r_{\rm
		inv}(\tau,Q)<h(\mathcal{A},G)+\varepsilon$.
	Let $\mathcal{S}$ be
	$(\tau,Q)$-spanning set that satisfies $r_{\rm inv}(\tau,Q)=N(\mathcal{S})$.
	It follows from Lemma~\ref{l:datarate:upperbound} that there exists a
	$Q$-admissible coder-controller $H$ so that $\frac{1}{\tau} \log_2
	N(\mathcal{S})\ge R(H)$ holds, and hence, we obtain $2\varepsilon+h_{\rm inv}{\color{black}(Q)}\ge
	R(H)$. This inequality holds for any $\varepsilon>0$, which implies that
	$h_{\rm inv}{\color{black}(Q)}\ge \inf_{H\in \mathcal{H}} R(H)$.
\end{proof}

\section{Uncertain Linear  Control Systems}
\label{s:lin}

We derive a lower bound of the invariance feedback entropy of uncertain linear
control systems \eqref{e:sys:lin} and compact sets $Q$. In this setting, we also derive
a lower bound of the data rate of any \emph{static} or \emph{memoryless}
coder-controller.
{\color{black} Similar to \cite[Section II]{NairFagniniZampieriEvans07} we employ the Brunn-Minkowsky
	inequality to obtain a lower bound on the growth of the size of the uncertainty
	set of the state at the controller side in one time step. For the
	general case, we use this inequality to derive a lower bound on the expansion
	number, which in turn leads to the entropy. For static coder-controllers the
	derivation of the lower bound is substantially simpler, see the proof of
	\cite[Thm~1]{NairFagniniZampieriEvans07} and the proof of Theorem \ref{t:lb:static}.}

\subsection{Universal lower bound}

\begin{theorem}\label{t:lb}
	Consider the matrices $A\in\R^{n\times n}$, $B\in\R^{n\times m}$ and two nonempty
	sets  $W,Q\subseteq \R^n$ with $W\subseteq Q$ and 
	%{$\mu(W) < \mu(Q)$ and} 
	suppose that $W$ is measurable and $Q$ is compact.
	Let
	%~\eqref{e:sys}
	$\Sigma$
	be given by $X=\R^n$, $U\subseteq \R^m$ with $U\neq
	\emptyset$
	and $F$ according to
	%\eqref{e:lb:lin}.
	\begin{equation}\label{e:lb:lin}
	\forall_{x\in X}\forall_{u\in U}\qquad
	F(x,u)=Ax+Bu+W.
	\end{equation}
	%Let $Q$ be controlled invariant with respect to $(X,U,F)$.
	{\color{black}Let $\R^n = \mathbb{E}_1 \oplus \mathbb{E}_2$, where 
		$\E_1$ is an $A$ invariant subspace of $\R^n$ with $\E_1 \neq \{0\}$, and $\oplus$ stands for the direct sum. 
		Let $\pi_1 : \R^n \to \E_1$ be the projection onto 
		$\E_1$ along $\E_2$, and\footnote{Since map $\pi_1$ is linear, we use notation $\pi_1 A$ instead of $\pi_1(A)$, $\forall A\subseteq\R^n$, for the sake of simpler presentation.} $\mu_1(\pi_1 W) < \mu_1(\pi_1Q)$, also let
		$n_1 = \texttt{dim}(\E_1)$ and $\mu_1$ denote the $n_1$-dimensional Lebesgue measure.
		%		 on $\E_1$. 
	}
	Then, the invariance feedback entropy of 
	%\eqref{e:sys} 
	$\Sigma$
	and $Q$ satisfies
	\begin{equation}\label{e:lb}
	\log_2\left( \vert\det
	A{{\color{black}|_{\E_1}}}\vert\frac{\mu_{\color{black}1}({\color{black}\pi_1} Q)}{(\mu_{\color{black}1}({\color{black}\pi_1} Q)^{\nicefrac{1}{n_{\color{black}1}}}-\mu_{\color{black}1}({\color{black}\pi_1} W)^{\nicefrac{1}{n_{\color{black}1}}})^{n_{\color{black}1}}}\right) \\
	\le
	h_{\rm inv}{\color{black}(Q)}.
	\end{equation}
\end{theorem}
\begin{proof}%[Proof of Theorem~\ref{t:lb}]
	
	Let us first point out that every compact set has finite Lebesgue measure.
	% and we
	%have $0\le\mu(W)^{\nicefrac{1}{n}}$  ${\color{black}<} \mu(Q)^{\nicefrac{1}{n}}$.
	%Therefore,
	%$1\le \mu(Q)^{\nicefrac{1}{n}}/(\mu(Q)^{\nicefrac{1}{n}}-\mu(W)^{\nicefrac{1}{n}})$ 
	%${\color{black}<}\infty$
	%and the left-hand-side
	%of~\eqref{e:lb} is well-defined. 
	
	If $|\det A{\color{black}|_{\E_1}}|=0$ the left-hand-side is $-\infty$
	and~\eqref{e:lb} holds.
	In the remainder we consider the case $|\det A{\color{black}|_{\E_1}}|>0$.
	If $h_{\rm inv}{\color{black}(Q)}=\infty$ the inequality~\eqref{e:lb} holds independent of the
	left-hand-side and subsequently we assume that $h_{\rm inv}{\color{black}(Q)}<\infty$. We
	pick $\varepsilon\in\R_{>0}$ and an invariant cover $(\mathcal{C},H)$ of
	%~\eqref{e:sys} 
	$\Sigma$
	and $Q$, so that
	$h(\mathcal{C},H)\le h_{\rm inv}{\color{black}(Q)}+\varepsilon$. Given Theorem~\ref{t:closed}, we
	can assume that the cover elements of $\mathcal{C}$ are closed, which yields by
	the compactness of $Q$ that the cover elements are compact and therefore
	Lebesgue measurable.
	
	We fix $\tau\in\N$ and pick a $(\tau,Q)$-spanning set $\mathcal{S}$ so that
	$r_{\rm inv}(\tau,Q)=N(\mathcal{S})$, which exists, since for fixed $\tau$, the number of $(\tau,Q)$-spanning set is finite.
	
	We are going to show that there exists
	$\alpha\in \mathcal{S}$  that satisfies
	\begin{equation}\label{e:lb:proof:1}
	\left(|\det
	A{\color{black}|_{\E_1}}|\frac{\mu_{\color{black}1}({\color{black}\pi_1} Q)}{(\mu_{\color{black}1}({\color{black}\pi_1} Q)^{\nicefrac{1}{n_1}}-\mu_{\color{black}1}({\color{black}\pi_1} W)^{\nicefrac{1}{n_1}})^{n_1}}\right)^\tau 
	\le 
	\prod_{t=0}^{\tau-1} \no P(\alpha|_{\intcc{0;t}}).
	\end{equation}
	We construct $\alpha\in\mathcal{S}$ iteratively over $t\in\intco{0;\tau}$. For
	$t=0$ we introduce $S_0:=\{\alpha(0)\mid \alpha\in\mathcal{S}\}$ and define 
	\begin{equation*}
	m_0:=\max\{\mu_{\color{black}1}({\color{black}\pi_1}\alpha(0))^{\nicefrac{1}{n_1}}\mid \alpha\in \mathcal{S}\}.
	\end{equation*}
	We pick
	$\Omega_0\in S_0$ so that $m_0=\mu_{\color{black}1}({\color{black}\pi_1} \Omega_0)^{\nicefrac{1}{n_1}}$. For
	$t\in\intco{1;\tau-1}$ we set $\alpha_{t'}:=\Omega_0\cdots \Omega_{t'}$ for
	$t'\in\intcc{0;t}$ and assume
	that 
	$\Omega_{t'}\in P(\alpha|_{\intco{0;t'}})$
	and
	$\mu_{\color{black}1}({\color{black}\pi_1} \Omega_{t'})^{\nicefrac{1}{n_1}}=m_{t'}$
	holds for all $t'\in\intcc{1;t}$ where
	$$m_{t'}:=\max\{\mu_{\color{black}1}({\color{black}\pi_1} \Omega)^{\nicefrac{1}{n_1}}\mid \Omega\in P(\alpha|_{\intco{0;t'}})\}.$$ 
	Then we set 
	$m_{t+1}:=\max\{\mu_{\color{black}1}({\color{black}\pi_1} \Omega)^{\nicefrac{1}{n_1}}\mid \Omega\in
	P(\alpha|_{\intco{0;t+1}})\}$ and pick $\Omega_{t+1}\in
	P(\alpha|_{\intco{0;t+1}})$ so that $m_{t+1}=\mu_{\color{black}1}({\color{black}\pi_1} \Omega_{t+1})^{\nicefrac{1}{n_1}}$. For
	$t=\tau-1$ we obtain a sequence $\alpha:=\Omega_0\cdots\Omega_{\tau-1}$ that is
	an element of $\mathcal{S}$. Hence, it follows from~\eqref{e:icover:spanning} that $\alpha$ satisfies for all
	$t\in\intco{0;\tau}$ the inclusion 
	\begin{equation}\label{e:lb:proof:11}
	{\color{black}\pi_1} \big( \textstyle
	A\alpha(t)+BH(\alpha(t))+W \big) \subseteq {\color{black}\pi_1} \left( \bigcup_{\Omega\in
		P(\alpha|_{\intcc{0;t}})}\Omega \right).
	\end{equation}
	For $t\in\intco{0;\tau-1}$, we 
	use the  Brunn-Minkowsky inequality for compact, measurable
	sets~\cite{henstock1953measure} 
	\begin{equation*}
	\begin{split}
	\mu_{\color{black}1}({\color{black}\pi_1} A\alpha(t))^{\nicefrac{1}{n_1}}&+\mu_{\color{black}1}({\color{black}\pi_1} W)^{\nicefrac{1}{n_1}} 
	\le
	\mu_{\color{black}1}({\color{black}\pi_1} A\alpha(t)+{\color{black}\pi_1} BH(\alpha(t))+{\color{black}\pi_1} W)^{\nicefrac{1}{n_1}}
	\end{split}
	\end{equation*}
	and the equality \cite{tao2011introduction}
	\begin{equation*}
	\mu( A\alpha(t))^{\nicefrac{1}{n}}=|\det A|^{\nicefrac{1}{n}}\mu(\alpha(t))^{\nicefrac{1}{n}}
	\end{equation*}
	together with $\mu_{\color{black}1}({\color{black}\pi_1}\alpha(t))^{\nicefrac{1}{n_1}}=m_t$ and~\eqref{e:lb:proof:11}, to
	derive 
	\begin{equation}\label{e:lb:proof:7}
	|\det A{\color{black}|_{\E_1}}|^{\nicefrac{1}{n_1}}m_t+\mu_{\color{black}1}({\color{black}\pi_1} W)^{\nicefrac{1}{n_1}}
	\le
	m_{t+1}(\no P(\alpha|_{\intco{0;t+1}}))^{\nicefrac{1}{n_1}}
	\end{equation}
	for all $t\in\intco{0;\tau-1}$. Note that we also used the fact that $\E_1$ is $A$ invariant to show inequality \eqref{e:lb:proof:7}. Also, for every $t\in\intco{0;\tau}$ we have
	\begin{equation}\label{e:lb:proof:6}
	|\det A{\color{black}|_{\E_1}}|^{\nicefrac{1}{n_1}}m_t+\mu_{\color{black}1}({\color{black}\pi_1} W)^{\nicefrac{1}{n_1}}\le \mu_{\color{black}1}({\color{black}\pi_1} Q)^{\nicefrac{1}{n_1}}
	\end{equation}
	since $A\alpha(t)+BH(\alpha(t))+W\subseteq Q$ which follows from the
	fact that $\alpha(t)\in \mathcal{C}$ and $(\mathcal{C},H)$ is an invariant
	cover. To ease the notation, let us introduce $N_0:=(\no
	P(\alpha))^{\nicefrac{1}{n_1}}$ and 
	$N_t:=(\no P(\alpha|_{\intco{0;t}}))^{\nicefrac{1}{n_1}}$ for $t\in\intco{1;\tau}$.
	We use induction over $\tau'\in\intco{0;\tau}$ to show 
	\begin{equation}\label{e:lb:proof:5}
	\left(|\det A{\color{black}|_{\E_1}}|^{\nicefrac{1}{n_1}}\frac{\mu_{\color{black}1}({\color{black}\pi_1} Q)^{\nicefrac{1}{n_1}}}{\mu_{\color{black}1}({\color{black}\pi_1} Q)^{\nicefrac{1}{n_1}}-\mu_{\color{black}1}({\color{black}\pi_1} W)^{\nicefrac{1}{n_1}}}\right)^{\tau'+1}
	\le
	\prod_{t=0}^{\tau'} N_t.
	\end{equation}
	Let us show~\eqref{e:lb:proof:5} for {\color{black}$\tau'=0$}.
	Since $P(\alpha)$ is a cover of $Q$ and $\no P(\alpha)^{\nicefrac{1}{n}}=N_0$ we obtain
	\begin{equation}\label{e:lb:proof:4}
	\mu_{\color{black}1}({\color{black}\pi_1} Q)^{\nicefrac{1}{n_1}}\le m_0N_0.
	\end{equation}
	From \eqref{e:lb:proof:6} we obtain $m_0\le
	(\mu_{\color{black}1}({\color{black}\pi_1} Q)^{\nicefrac{1}{n_1}}-\mu_{\color{black}1}({\color{black}\pi_1} W)^{\nicefrac{1}{n_1}})/|\det A{\color{black}|_{\E_1}}|^{\nicefrac{1}{n_1}}$ and \eqref{e:lb:proof:5} follows for $\tau'=1$. 
	
	If $\tau=1$ we have
	shown~\eqref{e:lb:proof:5} and subsequently we consider $\tau>1$.
	We fix {\color{black}$\tau''\in\intco{1;\tau}$} and
	assume that \eqref{e:lb:proof:5} holds for all $\tau'\in \intco{0;\tau''}$. We
	use \eqref{e:lb:proof:7} {\color{black}recursively} to derive
	\begin{equation}\label{e:lb:proof:8}
	m_0
	\le
	\frac{m_{\tau''}}{|\det A{\color{black}|_{\E_1}}|^{\tau''/n_1}}\bigg(\prod_{t=1}^{\tau''}N_t\bigg)
	-
	\sum_{t=1}^{\tau''}\frac{\mu_{\color{black}1}({\color{black}\pi_1} W)^{\nicefrac{1}{n_1}}}{|\det A{\color{black}|_{\E_1}}|^{t/n_1}}\prod_{t'=1}^{t-1}N_{t'}
	\end{equation}
	with the convention that $\prod_{t=a}^b x_t = 1$ for $b<a$.
	Using \eqref{e:lb:proof:4} and rearranging the terms in \eqref{e:lb:proof:8} we obtain
	\begin{equation}
	\mu_{\color{black}1}({\color{black}\pi_1} Q)^{\nicefrac{1}{n_1}}
	+
	\sum_{t=1}^{\tau''}\frac{\mu_{\color{black}1}({\color{black}\pi_1} W)^{\nicefrac{1}{n_1}}}{|\det A{\color{black}|_{\E_1}}|^{t/n_1}}\prod_{t'=0}^{t-1}N_{t'} 
	\le
	\frac{m_{\tau''}}{|\det A{\color{black}|_{\E_1}}|^{\tau''/n_1}}\prod_{t=0}^{\tau''}N_t.
	\end{equation}
	We invoke the induction hypothesis and use the inequality \\
	$\prod_{t'=0}^{t-1}N_{t'} \ge ((|\det A{\color{black}|_{\E_1}}|\mu_{\color{black}1}({\color{black}\pi_1} Q))^{\nicefrac{1}{n_1}}/(\mu_{\color{black}1}({\color{black}\pi_1} Q)^{\nicefrac{1}{n_1}}-\mu_{\color{black}1}({\color{black}\pi_1} W)^{\nicefrac{1}{n_1}}))^t$ to derive
	\begin{equation}\label{e:lb:proof:10}
	\mu_{\color{black}1}({\color{black}\pi_1} Q)^{\nicefrac{1}{n_1}}
	+
	\sum_{t=1}^{\tau''}\frac{\mu_{\color{black}1}({\color{black}\pi_1} W)^{\nicefrac{1}{n_1}}\mu_{\color{black}1}({\color{black}\pi_1} Q)^{\nicefrac{t}{n_1}}}{(\mu_{\color{black}1}({\color{black}\pi_1} Q)^{\nicefrac{1}{n_1}}-\mu_{\color{black}1}({\color{black}\pi_1} W)^{\nicefrac{1}{n_1}})^t} 
	\le
	\frac{m_{\tau''}}{|\det A{\color{black}|_{\E_1}}|^{\tau''/n_1}}\prod_{t=0}^{\tau''}N_t.
	\end{equation}
	From Lemma~\ref{l:1} (given in the Appendix) it follows that
	the left-hand-side of \eqref{e:lb:proof:10} evaluates to
	%\begin{multline}\label{e:lb:proof:9}
	\begin{equation}\label{e:lb:proof:9}
	\mu_{\color{black}1}({\color{black}\pi_1} Q)^{\nicefrac{1}{n_1}}
	+
	\sum_{t=1}^{\tau''}\frac{\mu_{\color{black}1}({\color{black}\pi_1} W)^{\nicefrac{1}{n_1}}\mu_{\color{black}1}({\color{black}\pi_1} Q)^{\nicefrac{t}{n_1}}}{(\mu_{\color{black}1}({\color{black}\pi_1} Q)^{\nicefrac{1}{n_1}}-\mu_{\color{black}1}({\color{black}\pi_1} W)^{\nicefrac{1}{n_1}})^t} 
	=
	\frac{\mu_{\color{black}1}({\color{black}\pi_1} Q)^{\nicefrac{(\tau''+1)}{n_1}}}{(\mu_{\color{black}1}({\color{black}\pi_1} Q)^{\nicefrac{1}{n_1}}-\mu_{\color{black}1}({\color{black}\pi_1} W)^{\nicefrac{1}{n_1}})^{\tau''}}.
	\end{equation}
	%\end{multline}
	We combine $m_{\tau''}\le (\mu_{\color{black}1}({\color{black}\pi_1} Q)^{\nicefrac{1}{n_1}}-\mu_{\color{black}1}({\color{black}\pi_1} W)^{\nicefrac{1}{n_1}})/|\det
	A{\color{black}|_{\E_1}}|^{\nicefrac{1}{n_1}}$ (that follows from \eqref{e:lb:proof:6})
	with \eqref{e:lb:proof:10} and \eqref{e:lb:proof:9} to get 
	\begin{equation}
	\frac{\mu_{\color{black}1}({\color{black}\pi_1} Q)^{\nicefrac{(\tau''+1)}{n_1}}}{(\mu_{\color{black}1}({\color{black}\pi_1} Q)^{\nicefrac{1}{n_1}}-\mu_{\color{black}1}({\color{black}\pi_1} W)^{\nicefrac{1}{n_1}})^{\tau''}} 
	\le
	\frac{\mu_{\color{black}1}({\color{black}\pi_1} Q)^{\nicefrac{1}{n_1}}-\mu_{\color{black}1}({\color{black}\pi_1} W)^{\nicefrac{1}{n_1}}}{|\det A{\color{black}|_{\E_1}}|^{(\tau''+1)/n_1}}\prod_{t=0}^{\tau''}N_t
	\end{equation}
	which shows that~\eqref{e:lb:proof:5} holds for {\color{black}$\tau'=\tau''$}. Hence,
	\eqref{e:lb:proof:5} holds for all  $\tau'\in\intco{0;\tau}$. In particular, for
	$\tau'=\tau-1$ and we conclude that 
	\eqref{e:lb:proof:1} holds.
	
	Inequality \eqref{e:lb:proof:1} together with the definition of
	$N(\mathcal{S})$ yields
	\begin{equation*}\label{e:lb:proof:2} 
	\begin{split}
	\left(|\det
	A{\color{black}|_{\E_1}}|\frac{\mu_{\color{black}1}({\color{black}\pi_1} Q)}{(\mu_{\color{black}1}({\color{black}\pi_1} Q)^{\nicefrac{1}{n_1}}-\mu_{\color{black}1}({\color{black}\pi_1} W)^{\nicefrac{1}{n_1}})^{n_1}}\right)^\tau\le
	N(\mathcal{S}) 
	=r_{\rm inv}(\tau,Q)
	\end{split}
	\end{equation*}
	where the equality follows by our choice of $\mathcal{S}$.
	From~\eqref{e:icover:entropy} we get 
	\begin{equation}\label{e:lb:proof:3} 
	\begin{split}
	\log_2
	\left(
	|\det A{\color{black}|_{\E_1}}|\frac{\mu_{\color{black}1}({\color{black}\pi_1} Q)}{(\mu_{\color{black}1}({\color{black}\pi_1} Q)^{\nicefrac{1}{n_1}}-\mu_{\color{black}1}({\color{black}\pi_1} W)^{\nicefrac{1}{n_1}})^{n_1}}
	\right) 
	\le
	h(\mathcal{C},H)\le h_{\rm inv}{\color{black}(Q)}+\varepsilon
	\end{split}
	\end{equation}
	which implies \eqref{e:lb} since \eqref{e:lb:proof:3} holds for every $\varepsilon>0$.
\end{proof}
{\color{black}
	\begin{remark}
		Let $\texttt{spec}(A)$ denote the spectrum of $A$, $\E^{\lambda}$ denote the eigenspace of $A$ associated with 
		$\lambda \in \texttt{spec}(A)$ and $B \subseteq \texttt{spec}(A)$.
		In Theorem~\ref{t:lb} if $\E_1 = \bigoplus_{\lambda \in B} \E^{\lambda}$, then a good choice of
		$\E_1$ will be the one that gives the largest lower bound in~\eqref{e:lb}. 
\end{remark} }

\begin{remark}
	Note that the lower bound, i.e., the left-hand-side of inequality~\eqref{e:lb}, is invariant under coordinate transformation. Let
	$z=Tx$ for some invertible matrix $T\in\R^{n\times n}$ so that the
	transition function $\bar F$ of the system in the new coordinates is 
	\begin{equation}\label{e:sys:lin1}
	\bar F(z,u)=TAT^{-1}z+TBu+TW
	\end{equation}
	and $\bar Q=TQ$. {\color{black}Let $\bar{\E}_i =T \E_i$, $i\in\{1,2\}$, $\bar{\pi}_1:\R^n\to \bar{\E}_1$ be the projection on $\bar{\E}_1$ along $\bar{\E}_2$. } Then we obtain %\MZ{the left-hand-side of inequality~\eqref{e:lb} is written based on projection on $\E_1$, whereas the one below is based on $\E_1=\R^n$!}
	\begin{align*}
	& |\det (TAT^{-1}){\color{black}|_{\bar{\E}_1}}|\frac{\mu_{\color{black}1}({\color{black}\bar{\pi}_1}TQ)}{(\mu_{\color{black}1}({\color{black}\bar{\pi}_1}TQ)^{\nicefrac{1}{n_{\color{black}1}}}-\mu_{\color{black}1}({\color{black}\bar{\pi}_1}TW)^{\nicefrac{1}{n_{\color{black}1}}})^{n_{\color{black}1}}}=\\
	& {\color{black}|\det A|_{{\E}_1}|\frac{\mu_1(T{\pi}_1Q)}{(\mu_1(T{\pi}_1Q)^{\nicefrac{1}{n_1}}-\mu_1(T{\pi}_1W)^{\nicefrac{1}{n_1}})^{n_1}}=}\\
	%  & |\det A{\color{black}|_{{\E}_1}}|\frac{|\det T{\color{black}|_{{\E}_1}}|\mu_{\color{black}1}({\color{black}\pi_1}Q)}{|\det T{\color{black}|_{{\E}_1}}|(\mu_{\color{black}1}({\color{black}\pi_1}Q)^{\nicefrac{1}{n_{\color{black}1}}}-\mu_{\color{black}1}({\color{black}\pi_1}W)^{\nicefrac{1}{n_{\color{black}1}}})^{n_{\color{black}1}}}=\\
	& |\det A{\color{black}|_{{\E}_1}}|\frac{\mu_{\color{black}1}({\color{black}\pi_1}Q)}{(\mu_{\color{black}1}({\color{black}\pi_1}Q)^{\nicefrac{1}{n_{\color{black}1}}}-\mu_{\color{black}1}({\color{black}\pi_1}W)^{\nicefrac{1}{n_{\color{black}1}}})^{n_{\color{black}1}}}.
	\end{align*}
\end{remark}

When 
%all eigenvalues of  $A$ are larger than one, and
$W$ is a singleton set,  {\color{black}by taking $\E_1$ as the unstable subspace, we get the largest lower bound in \eqref{e:lb} which} recovers the well-known
value of the invariance entropy \cite[Th.~3.1]{Kawan13} for deterministic linear control
systems, i.e., the invariance entropy equals $\log_2|\det A{\color{black}|_{\E_1}}|$. This matches
also other results known from stabilization with rate limited feedback~\cite{tatikonda2004control}.

\subsection{Static coder-controllers}

We restrict our attention to static coder-controllers and derive a lower bound
of the data rate of such coder-controllers.

Let $(\mathcal{C},H)$ be an invariant cover of \eqref{e:sys} and a nonempty set
$Q\subseteq X$.
We define the \emph{data rate} of $(\mathcal{C},H)$ by
\begin{equation}\label{s:static:dr}
R(\mathcal{C},H):=\log_2 \no\mathcal{C}.
\end{equation}
The definition is motivated by the fact that any invariant cover
$(\mathcal{C},H)$ immediately provides a \emph{static} or \emph{memoryless}
coder-controller scheme: given $x\in Q$ at the coder side, 
it is sufficient that the coder transmits one of the cover elements
$C\in\mathcal{C}$ that
contains the current state $x\in C$, to ensure that the controller is able to
confine the successor states of $x$ to $Q$, i.e.,
\begin{equation}
Ax+BH(C)+W\subseteq Q.
\end{equation}
%Given that the cover element that contains the initial state
%of~\eqref{e:sys:lin} is known to the controller at time zero,
The number of different cover elements that need to be transmitted via the
digital, noiseless channel at any time $t>0$ is bounded by $\no \mathcal{C}$.
Neither the coder nor the controller requires any past information for a correct
functioning.  Hence, we speak of $(\mathcal{C},H)$ as static or memoryless coder-controller for $(X,U,F)$.

The next result provides a lower bound on the data rate of any static
coder-controller.

\begin{theorem}\label{t:lb:static}
	Consider the matrices $A\in\R^{n\times n}$, $B\in\R^{n\times m}$ and two nonempty
	sets  $W,Q\subseteq \R^n$ with $W\subseteq Q$ and suppose that $W$ is measurable and $Q$ is compact.
	Let~\eqref{e:sys} be given by $X=\R^n$, $U\subseteq \R^m$ with $U\neq\emptyset$, 
	$F$ according to \eqref{e:lb:lin}, {\color{black}$\E_1$, $\E_2$, $\mu_1$, $n_1$ and $\pi_1$ as in Theorem~\ref{t:lb} and $\mu_1(\pi_1W) < \mu_1(\pi_1Q)$}.
	%Let $Q$ be controlled invariant with respect to $(X,U,F)$.
	Then, we have
	%the static invariance feedback entropy of \eqref{e:sys} and $Q$ satisfies
	\begin{equation}\label{e:lb:static}
	\log_2\left\lceil |\det
	A{\color{black}|_{\E_1}}|\frac{\mu_{\color{black}1}({\color{black}\pi_1}Q)}{(\mu_{\color{black}1}({\color{black}\pi_1}Q)^{\nicefrac{1}{n_{\color{black}1}}}-\mu_{\color{black}1}({\color{black}\pi_1}W)^{\nicefrac{1}{n_{\color{black}1}}})^{n_{\color{black}1}}}\right\rceil 
	\le   \inf_{(\mathcal{C},H)} R(\mathcal{C},H)
	\end{equation}
	where we take the infimum over all invariant covers $(\mathcal{C},H)$ of
	\eqref{e:sys} and $Q$. %\MZ{can we modify this result also based on the projection idea in the previous theorem?}
\end{theorem}
\begin{proof}
	%Every compact set has finite Lebesgue measure and from 
	% $W\subseteq Q$ it follows
	% $0\le\mu(W)^{\nicefrac{1}{n}}$ {\color{black}$<$}$\mu(Q)^{\nicefrac{1}{n}}$. Hence,
	%$1\le
	%\mu(Q)^{\nicefrac{1}{n}}/(\mu(Q)^{\nicefrac{1}{n}}-\mu(W)^{\nicefrac{1}{n}})${\color{black}$<$}$\infty$
	%and the left-hand-side
	%of~\eqref{e:lb:static} is well-defined. 
	If $|\det A{\color{black}|_{\E_1}}|=0$ the left-hand-side
	of~\eqref{e:lb:static} evaluates to $-\infty$ so that \eqref{e:lb:static} holds.
	Let us consider $|\det A{\color{black}|_{\E_1}}|>0$. 
	If the right-hand-side of~\eqref{e:lb:static} evaluates to $\infty$ nothing needs to be
	shown and we consider $\inf_{(\mathcal{C},H)} R(\mathcal{C},H)<\infty$. Since
	$\inf_{(\mathcal{C},H)} R(\mathcal{C},H)$ is finite, there exists an invariant cover
	$(\mathcal{D},G)$ of $(X,U,F)$ and $Q$. Let $(\mathcal{C},H)$ be the invariant cover with
	closed cover elements as constructed from $(\mathcal{D},G)$ in 
	{\color{black}Theorem~\ref{t:closed}}. 
	% ~\ref{t:closed} 
	{\color{black}Then} $(\mathcal{C},H)$ is an invariant cover 
	of $(X,U,F)$ and $Q$ and we have $R(\mathcal{C},H)$ {\color{black}$\leq$} $R(\mathcal D,G)$.
	
	As $(\mathcal{C},H)$ is an invariant cover of $(X,U,F)$ and $Q$, we have for
	every $\Omega\in \mathcal{C}$ the inclusion
	\begin{equation}
	{\color{black}\pi_1}(A\Omega+BH(\Omega)+W)\subseteq {\color{black}\pi_1}Q.
	\end{equation}
	We use the Brunn-Minkowsky inequality for compact, measurable
	sets {\color{black}(see proof of Theorem~\ref{t:lb})}
	%~\cite{henstock1953measure} 
	%$\mu(A\Omega)^{\nicefrac{1}{n}}+\mu(W)^{\nicefrac{1}{n}}\le 
	%\mu(A \Omega+BH(\Omega)+W)^{\nicefrac{1}{n}}$
	together with the identity \cite{tao2011introduction}
	$\mu(A\Omega)^{\nicefrac{1}{n}}=|\det A|^{\nicefrac{1}{n}}\mu(\Omega)^{\nicefrac{1}{n}}$
	to derive $|\det A{\color{black}|_{\E_1}}|^{\nicefrac{1}{n_{\color{black}1}}} \mu_{\color{black}1}({\color{black}\pi_1\Omega})^{\nicefrac{1}{n_{\color{black}1}}} +\mu_{\color{black}1}({\color{black}\pi_1}W)^{\nicefrac{1}{n_{\color{black}1}}}
	\le\mu_{\color{black}1}({\color{black}\pi_1}Q)^{\nicefrac{1}{n_{\color{black}1}}}$ which yields the bound 
	\begin{equation}\label{e:lb:static:proof:1}
	\mu_{\color{black}1}({\color{black}\pi_1}\Omega)^{\nicefrac{1}{n_{\color{black}1}}}\le
	\frac{\mu_{\color{black}1}({\color{black}\pi_1}Q)^{\nicefrac{1}{n_{\color{black}1}}}-\mu_{\color{black}1}({\color{black}\pi_1}W)^{\nicefrac{1}{n_{\color{black}1}}}}{|\det
		A{\color{black}|_{\E_1}}|^{\nicefrac{1}{n_{\color{black}1}}}}.
	\end{equation}
	As $\no \mathcal{C}$ is an upper bound
	on the number of cover elements needed to cover $F(\Omega,H(\Omega))$, we have 
	\begin{equation}\label{e:lb:static:proof:2}
	\mu_{\color{black}1}({\color{black}\pi_1}Q)^{\nicefrac{1}{n_{\color{black}1}}}\le 
	{\color{black}(\no \mathcal{C})^{\nicefrac{1}{n_{\color{black}1}}}}  
	\max \{\mu_{\color{black}1}({\color{black}\pi_1}\Omega)^{\nicefrac{1}{n_{\color{black}1}}}\mid \Omega\in \mathcal{C}\}.
	\end{equation}
	% \text{\color{black}$(\no \mathcal{C})^{\nicefrac{1}{n}}$}
	We use \eqref{e:lb:static:proof:1} (which holds for every
	$\Omega\in\mathcal{C}$) in \eqref{e:lb:static:proof:2} and rearrange the result
	to obtain
	\begin{equation*}
	|\det A{\color{black}|_{\E_1}}|^{\nicefrac{1}{n_{\color{black}1}}}\frac{\mu_{\color{black}1}({\color{black}\pi_1}Q)^{\nicefrac{1}{n_{\color{black}1}}}}
	{\mu_{\color{black}1}({\color{black}\pi_1}Q)^{\nicefrac{1}{n_{\color{black}1}}}-\mu_{\color{black}1}({\color{black}\pi_1}W)^{\nicefrac{1}{n_{\color{black}1}}}} 
	\le \color{black}(\no \mathcal{C})^{\nicefrac{1}{n_1}}.
	\end{equation*}
	Since this inequality holds for every invariant cover $(\mathcal{C},H)$, we
	obtain~\eqref{e:lb:static}.
	%\begin{multline*}
	%\log_2\left\lceil |\det A|_{\E_1}|\frac{\mu_1(\pi_1Q)}
	%{(\mu_1(\pi_1Q)^{\nicefrac{1}{n_1}}-\mu_1(\pi_1W)^{\nicefrac{1}{n_1}})^n_1}
	%\right\rceil \\
	%\le  \inf_{(\mathcal{C},H)}
	%R(\mathcal{C},H).
	%\qedhere
	%\end{multline*}
\end{proof}

It is easy to bound the difference between the universal lower bound in \eqref{e:lb} and the lower
bound of data rates for static coder-controllers in \eqref{e:lb:static} so that
we arrive at the following corollary, which allows us to quantify the
performance loss due to the restriction to static coder-controllers.

\begin{corollary}\label{c:loss}
	In the context and under the assumptions of Theorem~\ref{t:lb:static},
	let $a\in \R_{\ge0}$ be given by
	\begin{equation*}
	a:= |\det A{\color{black}|_{\E_1}}|\frac{\mu_{\color{black}1}({\color{black}\pi_1}Q)}{(\mu_{\color{black}1}({\color{black}\pi_1}Q)^{\nicefrac{1}{n_{\color{black}1}}}-\mu_{\color{black}1}({\color{black}\pi_1}W)^{\nicefrac{1}{n_{\color{black}1}}})^{n_{\color{black}1}}}.
	\end{equation*}
	%{\color{black}where $\E_1$ is such an $A$ invariant subspace that gives the largest value for the left hand side in~\eqref{e:lb:static}. }
	Suppose
	that $a<\infty$ and there exists an invariant cover $(\mathcal{C},H)$ of \eqref{e:sys} and $Q$ with
	$R(\mathcal{C},H)= \log_2 \ceil{ a }$. Then, the data rate $R$ of $(\mathcal{C},H)$ satisfies
	\begin{equation}\label{e:loss}
	R \le h_{\rm inv}{\color{black}(Q)}+1.
	\end{equation}
\end{corollary}
\begin{proof}
	Let $b\in \intco{0,1}$ be so that $a+b=\lceil a\rceil$.
	%Then we derive $R= \log_2(a+b)\le
	%\log_2(a(1+1/a))=\log_2(a)+\log_2(1+1/a)=h_{\rm inv}+\log_2(1+1/a)$. Since
	%$a\ge1$, we have $1+1/a\in \intcc{1,2}$ and $R\le h_{\rm inv}+1$ follows.
	We use $a\le
	2^{h_{\rm inv}{\color{black}(Q)}}$ and $0\le h_{\rm inv}{\color{black}(Q)}$ to derive 
	\begin{equation*}
	R=\log_2(a+b)\le\log_2(2^{h_{\rm inv}{\color{black}(Q)}}+b) 
	\le h_{\rm inv}{\color{black}(Q)}+\log_2(1+2^{-h_{\rm inv}{\color{black}(Q)}})	
	\le h_{\rm inv}{\color{black}(Q)}+1.\qedhere
	\end{equation*}
\end{proof}

\subsection{Tightness of the lower bounds}

We show for a particular class of scalar linear difference inclusions of the form
\begin{equation}\label{e:sys:lin:scalar}
\xi(t+1)\in a \xi(t)+\nu(t)+\intcc{w_1,w_2}
\end{equation}
with $a\in \R_{\neq 0}$, $w_1,w_2\in\R$ and $w_1\le w_2$ that the lower bounds
established in the previous subsections are tight.

Subsequently, we assume that $Q$ is given as
an interval containing $\intcc{w_1,w_2}$
\begin{equation*}
Q:=\intcc{q_1,q_2}, \qquad q_1,q_2\in\R, q_1<w_1,w_2<q_2.
\end{equation*}
We are going to construct a static coder-controller $(\mathcal{C},H)$ and show that its
data rate equals the lower bound in Theorem~\ref{t:lb:static}.%
\begin{subequations}
	\label{e:icover:scalar}
	To this end, we introduce
	\begin{equation}
	\begin{aligned}
	\Delta q &:=q_2-q_1,& &  &\Delta w&:=w_2-w_1,  \\
	q_c &:=(q_2 + q_1)/2&&\text{and}  &w_c&:=(w_2 + w_1)/2
	\end{aligned}
	\end{equation}
	and consider
	\begin{equation}\label{e:scalar:md}
	m:=\left\lceil |a|\frac{\Delta q}{ \Delta q-\Delta w}\right\rceil \text{ and } d:=\frac{\Delta q}{m}.
	\end{equation}
	Given $q_c$ and $d$, we introduce the intervals $\Lambda_i\subseteq\R$, $i\in\Z$
	\begin{equation}\label{e:icover:scalar:elements}
	\Lambda_i:= 
	\begin{cases}
	q_c+\intcc{id,(i+1)d}&\text{if $m$ is even} \\
	q_c+\intcc{(i-\tfrac{1}{2})d,(i+\tfrac{1}{2})d}&\text{if $m$ is odd}
	\end{cases}
	\end{equation}
	which we use to define 
	\begin{equation} 
	\mathcal{C}:=\{ \Lambda_i\cap Q\mid \Lambda_i\cap (\mathrm{int}Q)\neq\emptyset \}.
	\end{equation}
	The control function  follows for every
	$C_i\in\mathcal{C}$ by
	\begin{equation}\label{e:icover:scalar:input}
	H(C_i):=q_c-aq_c-w_c-
	\begin{cases}
	ad(i+\tfrac{1}{2})&\text{if $m$ is even} \\
	adi&\text{if $m$ is odd}.
	\end{cases}
	\end{equation}
\end{subequations}
For this construction of $(\mathcal{C},H)$, we have the following result.
\begin{theorem}\label{t:scalar:ub}
	Consider the scalars $a\in \R_{\neq 0}$, $w_1,q_1,w_2,q_2\in \R$ with $q_1<w_1\le w_2<q_2$.
	Let~\eqref{e:sys} be given by $X=U=\R$ and $F$ by
	$F(x,u)=ax+u+\intcc{w_1,w_2}$. Then, $(\mathcal{C},H)$ defined in
	\eqref{e:icover:scalar} is an invariant cover of \eqref{e:sys} and
	$\intcc{q_1,q_2}$ and we have
	\begin{equation}\label{e:scalar:datarate}
	\log_2\left\lceil |a|\frac{\Delta q}{ \Delta q-\Delta
		w}\right\rceil=R(\mathcal{C},H).
	\end{equation}
\end{theorem}
\begin{proof}
	We show the theorem for odd $m$. The case for even $m$, follows along the
	same arguments. It is rather straightforward to show that $\mathcal{C}$ is a
	cover of $Q$ and subsequently we show that $\no\mathcal{C}=m$. Note that
	$i>m/2-1/2$ implies that the left limit of $\Lambda_i$
	satisfies  $q_c+(i-\tfrac{1}{2})d \ge q_c+m/2d=q_2$, which shows that
	$i>m/2-1/2$ implies $\Lambda_i\cap(\mathrm{int}Q)=\emptyset$. Similarly,
	$i<-m/2+1/2$ implies $\Lambda_i\cap(\mathrm{int}Q)=\emptyset$, and we see that
	$\Lambda_i\cap(\mathrm{int}Q)\neq\emptyset$ implies $-m/2+1/2\le i\le
	m/2-1/2$ so that $\no\mathcal{C}\le m$ holds.

	We continue to show that 
	$F(C_i,H(C_i))\subseteq \intcc{q_1,q_2}$ holds for every $C_i\in\mathcal{C}$. Given
	\eqref{e:icover:scalar:input} we obtain for $F(C_i,H(C_i))$ the interval
	\begin{equation*}
	a((q_c+d\intcc{i-\tfrac{1}{2},i+\tfrac{1}{2}})\cap Q)+q_c-aq_c-w_c-adi+\intcc{w_1,w_2}\\
	\end{equation*}
	which is a subset of
	$I:=q_c+|a|\tfrac{d}{2}\intcc{-1,1}+\tfrac{\Delta w}{2}\intcc{-1,1}$.
	Let us show that  $I\subseteq Q$. Since $I$ is centered at $q_c$, it is
	sufficient to show $\nicefrac{|a|d}{2}+\nicefrac{\Delta w}{2}\le\nicefrac{\Delta
		q}{2}$. Note that $m\ge |a|\Delta q/(\Delta q-\Delta w)$ so that $d\le (\Delta
	q-\Delta w)/|a|$ follows and we obtain the desired inequality $\nicefrac{|a|d}{2}+\nicefrac{\Delta
		w}{2}\le \nicefrac{\Delta q}{2}$ which shows $F(C_i,H(C_i))\subseteq
	\intcc{q_1,q_2}$. Hence $(\mathcal{C},H)$ is an invariant cover with
	$R(\mathcal{C},H)\le \log_2 m$, which together with the inequality in
	Theorem~\ref{t:lb:static} shows the assertion.
\end{proof}

\begin{example}[continues=ex:linear]
	\normalfont
	Let us recall the linear system in Example~\ref{ex:linear} with
	$a=\nicefrac{1}{2}$, $W=\intcc{-3,3}$ and $Q=\intcc{-4,4}$. For this case,
	$m=2$ and $d=4$. The cover elements of $\mathcal{C}$ are given  according
	to~\eqref{e:icover:scalar:elements} by
	\begin{equation*}
	C_{-1}=\intcc{-4,0} \text{ and } C_{0}=\intcc{0,4}.
	\end{equation*}
	The inputs follow according to~\eqref{e:icover:scalar:input} by
	\begin{equation*}
	H(C_{-1})=1
	\text{ and }
	H(C_{0})=-1.
	\end{equation*}
	The data rate of $(\mathcal{C},H)$ is given by $\log_2
	2=1\;$bits per time unit.
\end{example}

We can use Corollary~\ref{c:loss} to conclude that the
performance loss due to the restriction to static coder-controllers in
Example~\ref{ex:linear} is no larger than $1$ bit/time unit. However, for this
example, and in general for scalar systems of the form~\eqref{e:sys:lin:scalar}
for which $|a|\Delta q/(\Delta q-\Delta w)$ is in $\N$, we see that the data rate of the proposed
static  coder-controller matches the best possible data rate $h_{\rm inv}{\color{black}(Q)}$ since in this case $R(\mathcal{C},H)$ equals the lower bound in Theorem~\ref{t:lb}.

The construction of static coder-controllers whose data rate  achieves the lower
bound in Theorem \ref{t:lb:static} in a more general setting is currently under investigation.

%\section{Conclusion}
%
%In this work we introduced a novel notion of invariance feedback entropy for
%uncertain control systems to characterize the critical data rate to achieve
%invariance. We used invariance feedback entropy to derive a universal
%  lower bound of the critical data rate of uncertain linear control systems. The
%  universal lower bound is intimately related to a lower bound of the data rate
%  of static, memoryless coder-controllers.
%  We showed that for certain linear control  systems the lower bounds are tight.
% 

{\color{black}
	\section*{Acknowledgement}
	The authors would like to thank Christoph Kawan for helpful discussions and
	suggestions to improve the manuscript, especially, for the suggestion of the subspace projection utilized in Theorem~\ref{t:lb}.
}

\printbibliography

\appendix
\section{}

{\color{black}
	\subsection{Mean-Payoff Games:} 
	A \emph{mean-payoff game} (MPG)~\cite{EhrenfeuchtMycielski79} is played by
	two players, player~$1$ and player~$2$, on a finite, directed, edge-weighted
	graph $G=(V,E,w)$, where
	$V:=V_1\cup V_2$, $V_1\cap V_2=\emptyset$ with $V_i$, $i\in \{1,2\}$ being two
	nonempty sets, $E\subseteq V\times V$, $w:E\to \Z$
	and for every $v\in V$ there exists $v'\in V$ so that $(v,v')\in E$. The
	vertices $V$ are also referred to as \emph{positions} of the game. Starting
	from an initial position $v_0\in V$, player~$1$ and player~$2$ take turns in picking the
	next position depending on the current position of the game: given $v_0\in V_i$ for $i\in\{1,2\}$ player~$i$ picks the
	successor vertex $v_1\in V$ so that $(v_0,v_1)\in E$ and the play
	continues with $v_1$. The infinite sequence of edges $e = (e_k)_{k\in \intco{0;\infty}}$ with 
	$e_k = (v_k,v_{k+1}) \in E$ is called a play.
	Player $1$ wants to minimize the payoff 
	\begin{equation*}
	\nu_{\mathrm{min}}(e_0e_1e_2\ldots):=\limsup_{k\to
		\infty}\frac{1}{k}\sum_{j=0}^{k-1} w(e_j)
	\end{equation*}
	while player $2$ wants to maximize the payoff 
	\begin{equation*}
	\nu_{\mathrm{max}}(e_0e_1e_2\ldots):=\liminf_{k\to
		\infty}\frac{1}{k}\sum_{j=0}^{k-1} w(e_j).
	\end{equation*}
	A \emph{positional strategy} for player $i$ is a function $\sigma_i:V_i\to V$ 
	so that $(v,\sigma_i(v))\in E$ holds for all $v\in V_i$.  
	By $\mathcal{P}_i(v,\sigma_i) \subseteq E^\intco{0;\infty}$ 
	we denote the set of all plays that start from the position $v$ 
	and wherein the player $i$ follows the positional strategy $\sigma_i$.
	
	As it turns out, there exist \emph{optimal
		positional strategies} $\sigma^*_i$ for each player~$i$ and a
	function $\nu:V\to \R$ so that player~$1$ is able to secure a payoff of $\nu(v)$
	against any other strategy of player~$2$ and vice versa, i.e., for all sequences
	$\check e\in\mathcal{P}_1(v,\sigma_1^*)$ and  $\hat
	e\in\mathcal{P}_2(v,\sigma_2^*)$ we have
	\begin{equation}\label{e:mpg:ineq}
	\nu_{\mathrm{min}}(\check e)
	\le 
	\nu(v)
	\le
	\nu_{\mathrm{max}}(\hat e).
	\end{equation}
	We call $\nu$ the \emph{value function} of the
	MPG $(V,E,w)$, see
	e.g.~\cite{EhrenfeuchtMycielski79} for details. Note
	that $\sigma_1^*$ is optimal in the sense that any deviation of player~$1$ from
	$\sigma_1^*$ can only lead to a larger or equal payoff than $\nu(v)$ considering
	the worst case with respect to possible strategies of player~$2$. Similarly,
	a deviation of player~$2$ from $\sigma_2^*$ may only lead to suboptimal payoff.
	{We exploit the following fact, which follows}
	from the proof of~\cite[Lemma.~1]{EhrenfeuchtMycielski79}: there exist constants $c_1$ and $c_2$, so that for every $\tau \in \N$, 
	$\check{e} \in \mathcal{P}_1(v,\sigma_1^*)$ and 
	$\hat{e} \in \mathcal{P}_2(v,\sigma_2^*)$ we have 
	\begin{equation}\label{eq:P1finiteSum}
	\frac{1}{\tau} \sum_{j=0}^{\tau-1}w(\check{e}_j) \leq \nu(v) + \frac{c_1}{\tau} 
	\end{equation}
	and 
	\begin{equation}\label{eq:P2finiteSum}
	\frac{1}{\tau} \sum_{j=0}^{\tau-1}w(\hat{e}_j) \geq \nu(v) + \frac{c_2}{\tau}.
	\end{equation}
	
	{We use the following lemma in Theorem \ref{t:frr} and Theorem
		\ref{t:closed}.}

	\begin{lemma}\label{lem:MPG}
		Consider two systems $\Sigma_i = (X_i,U_i,F_i)$, $i\in \{1,2\}$, a map $r: U_2
		\to U_1$ and let $Q_i$ be nonempty subsets of~$X_i$. 
		Suppose that
		$M:\wp(X_2)\to \wp(X_1) $
		maps subsets of $X_2$ to subsets
		of $X_1$ and 
		satisfies for every $u\in U_2$ and $A_2,A_2'\subseteq Q_2$ the
		following conditions
		\begin{enumerate}
			\item
			\label{lem:MPG:cover}
			$M(Q_2) = Q_1$,
			\item
			\label{lem:MPG:monotone}
			$A_2\subseteq A_2'\implies M(A_2)\subseteq M(A_2')$,
			\item 
			\label{lem:MPG:subset}
			%$M(A_2 \cup A_2')\subseteq M(A_2)\cup M(A_2') $ and
			$M(A_2 \cup A_2')= M(A_2)\cup M(A_2') $ and
			\item 
			\label{lem:MPG:semconj}
			$F_1(M(A_2),r(u))\subseteq M(F_2(A_2,u))$.
		\end{enumerate}
		%	{\color{red} We assume that}
		%	\begin{enumerate}
		%		\item \label{lem:MPG:A1covers} 
		%		$\mathcal{A}_1 = \left\{M(A)\mid A\in \mathcal{A}_2\right\}$ is a cover of
		%    $Q_1$ and
		%	\end{enumerate}
		Let $(\mathcal{A}_2,G_2)$ be an invariant cover of $\Sigma_2$ and~$Q_2$ and
		let $$\mathcal{A}_1: = \left\{M(A)\mid A\in \mathcal{A}_2\right\}.$$
		%
		%	For all $\tau\in \N$, $t\in \intco{0;\tau-1}$, if $\mathcal{S}_2=\mathcal{A}_2^\intco{0;\tau} $ is a $(\tau,Q)$-spanning set in $(\mathcal{A}_2,G_2)$ for $\Sigma_2$,
		%	and 
		%	for $\alpha\in \mathcal{S}_2$, if 
		%	$V'=\left\{M(A)\mid A\in P(\alpha|_\intcc{0;t})\right\}$, 
		%	$u =G_2(\alpha(t))$, $A_1=M(\alpha(t))$, then we assume that
		%	\begin{enumerate}[resume]
		%		\item \label{lem:MPG:seqSucCovPos} $F_1(A_1,r(u))\subseteq \cup_{A \in V'} A$.
		%		%\item \label{lem:MPG:semconj}Further assume that $F_1(M(A),r(u))\subseteq M(F_2(A,u))$ for every $u\in U_2$, $A\in \mathcal{A}_2$.
		%	\end{enumerate}
		Then there exists a map $G_1^*:\mathcal{A}_1 \to U_1$ such that $(\mathcal{A}_1,G_1^*)$ is an invariant cover of $\Sigma_1$ and $Q_1$, and
		\begin{equation}
		h(\mathcal{A}_1,G_1^*) \leq h(\mathcal{A}_2,G_2).
		\end{equation}
	\end{lemma}
	
	\begin{proof}
		{ Let us first point out that $\mathcal{A}_1$ is a
			cover of $Q_1$. We use \ref{lem:MPG:cover}) and
			\ref{lem:MPG:subset}) to derive
			\begin{equation*}
			Q_1=M(Q_2)=M(\cup_{A_2\in\mathcal{A}_2}A_2)
			=
			\cup_{A_2\in\mathcal{A}_2}M(A_2)
			\end{equation*}
			and we see that $\mathcal{A}_1$ is
			a cover of $Q_1$.
			%  which shows $Q_1 = \cup_{A_1\in\mathcal{A}_1}A_1$ follows. Moreover, for
			%  every element $A_1\in \mathcal{A}_1$ there exists $A_2\in\mathcal{A}_2$ so
			%  that $A_1=M(A_2)$. Hence $A_1\subseteq Q_1$ .
		}
		
		{ Consider the map $G_1:\mathcal{A}_1\rightrightarrows U_1$ defined by 
			\begin{align*}
			G_1(A_1):= \left\{r(G_2(A_2))\mid A_2 \in \mathcal{A}_2, M(A_2) = A_1\right\}
			\end{align*}
			and let
			\begin{equation*} 
			\mathcal{V}(A_1):=
			\big\{ (V,u) \mid V \subseteq \mathcal{A}_1, u \in G_1(A_1),
			F_1(A_1,u) \subseteq \cup_{A \in V}A \big\}.
			\end{equation*} We show that
			$\mathcal{V}(A_1)$ is nonempty for every $A_1\in\mathcal{A}_1$. Let  $A_1\in\mathcal{A}_1$  and $u\in G_1(A_1)$. Then there exists
			$A_2\in\mathcal{A}_2$ so that $A_1=M(A_2)$ and $u=r(G_2(A_2))$. We use
			\ref{lem:MPG:semconj}) to see that 
			$F_1(A_1,u)\subseteq M(F_2(A_2,G_2(A_2)))$. Since $(\mathcal{A}_2,G_2)$ is
			an invariant cover we have $F_2(A_2,G_2(A_2))\subseteq Q_2$ and it follows
			from \ref{lem:MPG:monotone}) that $F_1(A_1,u)\subseteq M(Q_2)$. Since
			$\mathcal{A}_1$ covers $M(Q_2)=Q_1$, we see that $F_1(A_1,u) \subseteq
			\cup_{A \in \mathcal{A}_1}A$, which ensures that
			$\mathcal{V}(A_1)\neq\emptyset$.
			
			%    \subseteq 
			%    M(\cup_{A_2\in P(\alpha|_{\intcc{0;t}})}A_2)
			%    \subseteq 
			%    \cup_{A_2\in P(\alpha|_{\intcc{0;t}})}M(A_2)
			%  \end{multline*}
			
		}
		
		Given $\Sigma_1$ and $(\mathcal{A}_1,G_1)$ we construct an MPG $(V,E,w)$.
		Let $V_1:=\mathcal{A}_1$ and $V_2:=\cup_{A\in V_1}\mathcal{V}(A)$  then the \emph{positions}
		of the MPG follow by $V=V_1\cup V_2$. 
		We introduce the \emph{edges} $E:=E_1\cup E_2$ of the MPG by
		\begin{align*}
		E_1&:=\{(v_1,v_2)\in V_1\times V_2\mid   v_2\in \mathcal{V}(v_1)\}\\
		E_2&:=\{(v_2,v_1)\in V_2\times V_1\mid v_1 \in V',  v_2 = (V',u)    \}.
		\end{align*}
		For $v\in V_2$ with $v = (V',u)  $ by $\no v$ we refer to $\no V' $. 
		The weights for $(v_1,v_2)\in E_1$ and $(v_2,v_1)\in E_2$ are given by 
		$w(v_1,v_2):=\log_2\no v_2$ and $w(v_2,v_1):=\log_2\no v_2$. 
		We refer to $(V,E,w)$ as the MPG associated with 
		$\Sigma_1$ and $(\mathcal{A}_1,G_1)$.
		{ Subsequently, we use  $\sigma_i^*$, $i\in\{1,2\}$ to denote the optimal positional strategy for
			player~$i$.}

		%	Though in Section~\ref{s:inv} we used $r_{\text{inv}}(\tau, Q)$ to refer to the 
		%	smallest possible expansion number for any $(\tau,Q)$-spanning set in 
		%	an invariant cover $(\mathcal{A},G)$, in the following we use 
		%	$r_{\text{inv}}(\tau, \mathcal{A},G,Q)$ for clarity.
		%
		{ Fix $\tau \in \N$  and let
			$r_{2,\text{inv}}(\tau,Q_2)$ denote the 
			smallest possible expansion number associated with 
			the invariant cover $(\mathcal{A}_2,G_2)$ at time~$\tau$.
			%%
			%  {\color{red} Subsequently, for $i\in\{1,2\}$ we use
			%  $r_{i,\text{inv}}(\tau,Q_i)$ to refer to the 
			%	smallest possible expansion number associated with 
			%	the invariant covers $(\mathcal{A}_i,G_i)$ at time $\tau$.}
			%
			%
			Let $\mathcal{S}_2$ be a 
			$(\tau,Q)$-spanning set in $(\mathcal{A}_2,G_2)$ such that 
			$N(\mathcal{S}_2) = r_{2,\text{inv}}(\tau, Q_2)$.
			%Let $V_0 = \{ M(\alpha(0)) \mid \alpha \in \mathcal{S}_2 \}$ and pick
			%$\bar{v} \in V_0$ so that $\nu (\bar{v} ) = \max_{v \in V_0} \nu (v)$.
			%For $\mathcal{S}_{2,0} = \{ \alpha(0) \mid \alpha \in \mathcal{S}_2 \}$,
			%we observe that
			%$Q_1 = M(Q_2)$ $=M(\cup_{A\in \mathcal{S}_{2,0}}A)$ 
			%$\subseteq\cup_{A\in\mathcal{S}_{2,0}}M(A)$ $=\cup_{A\in V_0}A$. 
			We observe that
			$Q_1 = M(Q_2) =M(\cup_{\alpha \in \mathcal{S}_{2}}\alpha(0))\subseteq\cup_{\alpha \in \mathcal{S}_{2}}M(\alpha(0))$.
			Thus $V_0:=\{ M(\alpha(0)) \mid \alpha \in \mathcal{S}_2 \}$ covers $Q_1$. 
			We pick  $\bar{v} \in V_0$ so that $\nu (\bar{v} ) = \max_{v \in V_0} \nu (v)$.
			We show by induction over $t \in \intco{0;\tau-1}$ the existence of an $\alpha \in \mathcal{S}_2$ and an 
			\mbox{$(v_k,v_{k+1})_{k\in\intco{0;\infty}} \in \mathcal{P}_2(\bar{v},
				\sigma_2^*)$} such that 
			\begin{equation}
			\label{e:proof:lem:MPG:1}
			v_{2k} = M(\alpha(k)) \text{ and }
			v_{2k+1} = ( \{ M(A) \mid A \in P(\alpha|_{\intcc{0;k}})  \}, u_k  )
			\end{equation}
			with $u_k= r(G_2(\alpha(k)))$ holds 
			for all $k \in \intcc{0;t}$.
			Let $t=0$, then}
		there exists $\alpha \in \mathcal{S}_2$ with 
		$M(\alpha(0)) = \bar{v}$.
		As $\mathcal{S}_2$ is $(\tau,Q)$-spanning we have 
		$F_2(\alpha(0), G_2(\alpha(0)) ) \subseteq \cup_{A\in P(\alpha(0))} A$.
		{
			For $u=G_2(\alpha(0))$ and $V' = \{M( A) \mid A \in P(\alpha(0)) \} $ we
			use \ref{lem:MPG:semconj}), \ref{lem:MPG:monotone}) and \ref{lem:MPG:subset}) to derive
			\begin{equation}
			\label{e:proof:lem:MPG:2}
			F_1(\bar v,r(u))\subseteq M(F_2(\alpha(0),u))
			\subseteq 
			M(\cup_{A\in P(\alpha(0))}A)
			\subseteq 
			\cup_{A\in V'} A.
			\end{equation}%
			Hence, for $v_1:= (V', r(u))$ we have
			$v_1 \in \mathcal{V}(\bar{v})$ and $(\bar{v}, v_1) \in E_1$ thus 
			$e_0 = (\bar{v}, v_1) $ for some 
			$e \in  \mathcal{P}_2(\bar{v}, \sigma_2^*)$. 
			Now suppose that the induction hypothesis 
			\eqref{e:proof:lem:MPG:1}
			holds for { $t \in \intco{0;\tau-2}$,  $\alpha\in\mathcal{S}_2$ and
				$(v_k,v_{k+1})_{k\in\intco{0;\infty}}\in\mathcal{P}_2(\bar v,\sigma^\ast_2)$}.
			%Let $\hat{v}_{2t+2} = \sigma_2^*(v_{2t+1})$ then $(v_{2t+1}, \hat{v}_{2t+2}) \in E_2$. 
			Let $v_{2t+1} = (V',u)$. From the definition of $E_2$ we have 
			\mbox{$ v_{2t+2}=\sigma_2^\ast(v_{2t+1}) \in V'$}.
			Hence, together with \eqref{e:proof:lem:MPG:1} we see that there exists
			$A \in P(\alpha|_{\intcc{0;t}})$ with $M(A) = v_{2t+2}$.
			Then we can pick $\hat{\alpha} \in \mathcal{S}_{2}$ such that 
			%$A \in P(\alpha|_{\intcc{0;t}})$ with $M(A) = \hat{v}_{2t+2}$ we have
			$\hat{\alpha}|_{\intcc{0;t}} = \alpha|_{\intcc{0;t}}$ and $\hat{\alpha}(t+1) = A$.
			Further let $\hat{v}_{2t+3} = (V', r(u))$ with 
			$u = G_2(\hat{\alpha}(t+1) )$ and
			$V' = \{M(A) \mid A\in P(\hat{\alpha}|_{\intcc{0;t+1}} ) \}$.
			Then, by using the same arguments used to derive
			\eqref{e:proof:lem:MPG:2} with $v_{2t+2}$ and
			$P(\hat{\alpha}|_{\intcc{0;t+1}})$ in place of $\bar v$ and $P(\alpha(0))$ we
			obtain 
			\mbox{$F_1(v_{2t+2},r(u))\subseteq  \cup_{A\in V'}A$}.
			Thus 
			$( v_{2t+2}, \hat{v}_{2t+3} ) \in E_1$ and 
			there exists $e \in \mathcal{P}_2(\bar{v}, \sigma_2^*)$ such that 
			$e_k = (v_k,v_{k+1})$ for all $k\in \intcc{0;2t+1}$ and
			$e_{2t+2} = (v_{2t+2}, \hat{v}_{2t+3})$ which completes the induction.
			Let $\alpha$ and $e:=(v_k,v_{k+1})_{k\in\intco{0;\infty}}$ satisfy 
			\eqref{e:proof:lem:MPG:1} for all $k\in\intco{0;\tau-1}$, which implies
			$ \no v_{2t+1} \leq \no P(\alpha | _ \intcc{0;t} ) $ 
			for every $t \in \intco{0;\tau-1}$.
			As $e \in \mathcal{P}_2(\bar{v} , \sigma^*_2)$ from~(\ref{eq:P2finiteSum}) we have
			\begin{equation}
			\label{eq:AG}
			\begin{split}
			&\nu(\bar{v} ) + \frac{c_2}{2 \tau} \leq \frac{1}{2 \tau} \sum_{j=0}^{2 \tau-1}w({e}_j)  \\
			& \leq   \frac{1}{\tau}  \sum_{t=0}^{\tau-1}\log_2 \no P(\alpha|_{\intcc{0;t}}) + \frac{1}{\tau}  \log_2 \no  v_{2\tau-1}  - \frac{1}{\tau} \log_2\no V_0	\\
			& \leq \frac{1}{\tau} \log_2 r_{2,\text{inv}}(\tau, Q_2) + \frac{\bar{c}_2}{\tau}
			\end{split}
			\end{equation}
			where $	\bar{c}_2 = \log_2  \max_{v\in V_2} \no v$.}
		
		{ 
			We define $G_1^*: \mathcal{A}_1\to U_1$ based on the value of $\sigma_1^*(A)$,
			i.e., $G_1^*(A) := u$ where $\sigma_1^*(A) = (V',u) $.
			For any $A_1\in \mathcal{A}_1$ and $u=G_1^*(A_1)$ there exists $A_2\in \mathcal{A}_2$ such that 
			$A_1 = M(A_2)$ and $u=r(G_2(A_2))$. Hence, we use \ref{lem:MPG:semconj}),
			\ref{lem:MPG:monotone}) and
			\ref{lem:MPG:cover}) to derive	$F_1(A_1,G_1^*(A_1)) \subseteq M(F_2(A_2,G_2(A_2)))
			\subseteq M(Q_2)= Q_1$.
			Thus
			$(\mathcal{A}_1,G_1^*)$ is an invariant cover of $\Sigma_1$ and
			$Q_1$. }
		
		Now consider the set 
		$\mathcal{S}_1\subseteq \mathcal{A}_1^{\intco{0;\tau}}$ implicitly defined by 
		$\alpha\in\mathcal{S}_1$ if and only if there exists
		$(v_k, v_{k+1})_{k\in\intco{0;\infty}}\in \mathcal{P}_1(v_0,\sigma^*_1)$ with
		$v_0\in V_0$ so that $\alpha(t)=v_{2t}$ holds for all
		$t\in\intco{0;\tau}$. 
		The set $\{\alpha(0)\mid \alpha\in\mathcal{S}_1 \}$ equals 
		$V_0$ therefore it covers $Q_1$. 
		Consider any $\alpha \in \mathcal{S}_1$ and a play 
		$(v_k, v_{k+1})_{k\in\intco{0;\infty}}\in \mathcal{P}_1(v_0,\sigma^*_1)$ such that 
		$\alpha(t)=v_{2t}$ holds for all $t\in\intco{0;\tau}$. 
		For $k \in \intco{0;\tau-1}$ if $v_{2k+1} = (V',u)$ then from the definition of 
		$\mathcal{S}_1$  we have $P(\alpha |_\intcc{0;k}) = V'$ 
		and from the definition of the MPG we have that $V'$ covers $F_{1} (v_{2k},u)$.
		Thus
		\begin{equation*}
		\forall_{\alpha \in \mathcal{S}_1} \forall_{t \in \intco{0;\tau-1}}F(\alpha(t),G_1^*(\alpha(t))) \subseteq \cup_{A'\in P(\alpha \mid_\intcc{0;t})}A'.
		\end{equation*}
		Therefore $\mathcal{S}_1$ is a $(\tau,Q)$-spanning set in $(\mathcal{A}_1,G_1^*)$. 
		Let $\alpha \in \mathcal{S}_1$ such that 
		$ \prod_{t=0}^{\tau-1} \no P(\alpha|_{\intcc{0;t}}) =  N(\mathcal{S}_1) $.
		Pick an $e \in \mathcal{P}_1(\alpha(0), \sigma_1^*)$ such that $\alpha(t)=v_{2t}$ holds for all
		$t\in\intco{0;\tau}$. Then from~(\ref{eq:P1finiteSum}) we have 
		\begin{equation}
		\label{eq:CH*}
		\begin{split}
		&\nu (v_0) + \frac{c_1}{2 \tau}  \geq \frac{1}{2 \tau} \sum_{j=0}^{2 \tau-1}w({e}_j) \\
		&=  \frac{1}{\tau}  \sum_{t=0}^{\tau-1}\log_2 \no P(\alpha|_{\intcc{0;t}}) + \frac{1}{\tau}  \log_2 \no v_{2\tau-1}  - \frac{1}{\tau} \log_2\no V_0 	\\ 
		&\geq \frac{1}{\tau} \log_2 r_{1,\text{inv}} (\tau, Q_1) + \frac{\bar{c}_1}{\tau} 
		\end{split}
		\end{equation}
		where $\bar{c}_1 =   -  \log_2\no V_1 $.
		
		From (\ref{eq:AG}) and (\ref{eq:CH*}) we get
		\begin{equation*}
		\begin{split}
		&\frac{1}{\tau} \log_2 r_{1,\text{inv}} (\tau, Q_1) + \frac{\bar{c}_1}{\tau} 	
		\leq \nu (v_0 ) + \frac{c_1}{2\tau} 	
		{ \leq \nu (\bar v ) + \frac{c_1}{2\tau}} 	\\
		&\leq \frac{1}{\tau} \log_2 r_{2,\text{inv}}(\tau, Q_2) +
		{\frac{c_1+2\bar c_2- c_2}{2\tau}  }.
		\end{split}
		\end{equation*}
		{Since this inequality} holds for every $\tau \in \N$, {
			we get} 
		\begin{equation*}
		h(\mathcal{A}_1, G_1^*) \leq h(\mathcal{A}_2,G_2).
		\qedhere
		\end{equation*}
\end{proof} }

\subsection{Other Lemmas and Proofs}

\begin{proof}[Proof of Lemma~\ref{l:subadditivity}]
	We fix $\tau_1,\tau_2\in\N$ and choose two  minimal
	$(\tau_i,Q)$-spanning sets $\mathcal{S}_i$, $i\in\{1,2\}$ in $(\mathcal{A},G)$
	so that $r_{\rm inv}(\tau_i,Q)=N(\mathcal{S}_i$). 
	Let $\mathcal{S}$ be the set of sequences 
	$\alpha:\intco{0;\tau_1+\tau_2}\to \mathcal{A}$ given by
	$\alpha(t):=\alpha_1(t)$ for $t\in\intco{0;\tau_1}$ and
	$\alpha(t):=\alpha_2(t-\tau_1)$ for $t\in\intco{\tau_1;\tau_1+\tau_2}$, where
	$\alpha_i\in\mathcal{S}_i$ for $i\in\{1,2\}$.
	We claim that $\mathcal{S}$ is $(\tau_1+\tau_2,Q)$-spanning in
	$(\mathcal{A},G)$. It is easy to see that 
	$\{A\in\mathcal{A}\mid \exists_{\alpha\in\mathcal{S}} A= \alpha(0)\}$ covers
	$Q$, since
	$\{A\in\mathcal{A}\mid \exists_{\alpha\in\mathcal{S}_1} A= \alpha(0)\}$
	covers~$Q$. Let $t\in\intco{0;\tau_1+\tau_2}$ and $\alpha\in\mathcal{S}$. If
	$t\in\intco{0;\tau_1-1}$, we immediately see that
	$F(\alpha(t),G(\alpha(t)))\subseteq \cup_{A'\in P(\alpha|_{\intcc{0;t}})} A'$
	since $\alpha_1:=\alpha|_{\intco{0;\tau_1}}\in \mathcal{S}_1$  and
	$\mathcal{S}_1$ satisfies~\eqref{e:icover:spanning}. Similarly, if
	$t\in\intco{\tau_1;\tau_1+\tau_2-1}$, we have $F(\alpha(t),G(\alpha(t)))\subseteq
	\cup_{A'\in P(\alpha|_{\intcc{0;t}})} A'$ since 
	$\alpha_2:=\alpha|_{\intco{\tau_1;\tau_1+\tau_2}}\in \mathcal{S}_2$  and
	$\mathcal{S}_2$ satisfies~\eqref{e:icover:spanning}. For $t=\tau_1-1$, we know that 
	$P(\alpha|_{\intco{0;\tau_1}})$ equals $\{A\mid
	\exists_{\alpha_2\in\mathcal{S}_2}\; \alpha_2(0)=A\}$ which covers $Q$ and the
	inclusion $F(\alpha(t),G(\alpha(t)))\subseteq
	\cup_{A'\in P(\alpha|_{\intcc{0;t}})} A'$ follows. Hence, $\mathcal{S}$
	satisfies~\eqref{e:icover:spanning} and we see that $\mathcal{S}$ is
	$(\tau,Q)$-spanning. Subsequently, for $i\in\{1,2\}$ and $\alpha\in
	\mathcal{S}_i$, $t\in\intco{0;\tau_i-1}$, let us use
	$P_i(\alpha|_{\intcc{0;t}}):=\{A\in\mathcal{A}\mid \exists_{\hat
		\alpha\in\mathcal{S}_i}\;\hat\alpha|_{\intcc{0;t}}=\alpha|_{\intcc{0;t}}\wedge
	A=\hat\alpha(t+1)\}$. Then we have
	$P(\alpha|_{\intcc{0;t}})=P_1(\alpha_1|_{\intcc{0;t}})$ with
	$\alpha_1:=\alpha|_{\intco{0;\tau_1}}$ if $t\in\intco{0;\tau_1-1}$ and 
	$P(\alpha|_{\intcc{0;t}})=P_2(\alpha_2|_{\intcc{0;t-\tau_1}})$ with
	$\alpha_2:=\alpha|_{\intco{\tau_1;\tau_1+\tau_2}}$ if
	$t\in\intco{\tau_1;\tau_1+\tau_2-1}$, while for $t=\tau_1-1$ we have 
	$P(\alpha|_{\intcc{0;t}})=P_2(\alpha_2)$ with
	$\alpha_2:=\alpha|_{\intco{\tau_1;\tau_1+\tau_2}}$ and
	$P(\alpha):=P_1(\alpha_1)$ with $\alpha_1:=\alpha|_{\intco{0;\tau_1}}$.
	Therefore, $N(\mathcal{S})$ is bounded by $N(\mathcal{S}_1)\cdot N(\mathcal{S}_2)$ and we have $r_{\rm inv}(\tau_1+\tau_2,Q)\le r_{\rm inv}(\tau_1,Q)\cdot
	r_{\rm inv}(\tau_2,Q)$.
	Hence, $\tau\mapsto \log_2 r_{\rm inv}(\tau,Q)$, 
	$\N\to \R_{\ge0}$ is a subadditive
	sequence of real numbers and~\eqref{e:ifb:equiv} follows
	by~\cite[Lem.~2.1]{ColoniusKawanNair13}.
\end{proof}

\begin{proof}[Proof of Lemma~\ref{l:nospanningset}]
	For every $t\in\intco{0;\tau}$, we define the set $\mathcal{S}_t:=\{\alpha\in \mathcal{A}^{\intcc{0;t}}\mid
	\exists_{\alpha'\in\mathcal{S}}\;\alpha'|_{\intcc{0;t}}=\alpha\}$.
	By definition of $P$, we have for all $\alpha\in\mathcal{S}$ the equality
	$P(\alpha)=\mathcal{S}_0$, which shows the assertion for $\tau=1$ since in
	this case we have $\mathcal{S}_0=\mathcal{S}$. Subsequently, we assume $\tau>1$.
	For $t\in\intco{0;\tau}$ and $a_0\ldots a_{t}\in \mathcal{S}_t$, we use 
	$Y(a_0\ldots a_{t}):=\{ \alpha\in \mathcal{S} \mid a_0\ldots a_{t}= \alpha|_{\intcc{0;t}}\}$
	to denote the sequences in $\mathcal{S}$ whose initial part is restricted
	to $a_0\ldots a_{t}$. For 
	$t\in\intco{0;\tau-1}$ and $a_0\ldots a_{t}\in\mathcal{S}_{t}$, we have 
	\begin{align*}
	{\color{black}\no Y(a_0\ldots a_{t})} &{\color{black}= \sum_{a_{t+1} \in P(a_0 \ldots a_t)} \no Y(a_0\ldots a_{t+1})}	\\
	&\le \no P(a_0\ldots a_{t}) \max_{a_{t+1}\in P(a_0\ldots
		a_{t})}\no Y(a_0\ldots a_{t+1}).
	\end{align*}
	For every $a_0\ldots a_{\tau-2}\in\mathcal{S}_{\tau-2}$ we have 
	$\no Y(a_0\ldots a_{\tau-2})=\no P(a_0\ldots a_{\tau-2})$
	and we obtain a bound for $\no Y(a_0)$ by
	%\begin{multline*}
	\begin{equation*}
	\no P(a_0)\!\max_{a_1\in P(a_0)}\!\no P(a_0a_1)\!\cdots\!
	\max_{a_{\tau-2}\in P(a_0\ldots a_{\tau-3})} \!\no P(a_0\ldots a_{\tau-2})
	\end{equation*}
	%\end{multline*}
	so that $\no Y(a_0)\le
	\max_{\alpha\in\mathcal{S}}\prod_{t=0}^{\tau-2}\no P(\alpha|_{\intcc{0;t}})$ holds for any $a_0\in \mathcal{S}_0$.
	{\color{black} As $\cup_{a_0\in \mathcal{S}_0} Y(a_0)=\mathcal{S}$ we observe
		$\no \mathcal{S} = \sum_{a_0 \in \mathcal{S}_0} \no Y(a_0) 
		\le \no \mathcal{S}_0 \max_{\alpha \in \mathcal{S}} \prod_{t=0}^{\tau-2}\no P(\alpha|_{\intcc{0;t}})$. 
		Since $\mathcal{S}_0 = P(\alpha) = P(\alpha|_\intcc{0;\tau-1})$, we obtain the
	}
	desired inequality
	$\no \mathcal{S}\le  \max_{\alpha\in\mathcal{S}}\prod_{t=0}^{\tau-1} \no
	P(\alpha|_{\intcc{0;t}})$.
\end{proof}

\begin{lemma}\label{l:1}
	For $a,b\in \R$ and $T\in \N$, it holds
	\begin{equation}\label{e:sumeq}
	a + \sum_{t=1}^{T}\frac{ba^t}{(a-b)^t}
	= 
	\frac{a^{T+1}}{(a-b)^{T}}.
	\end{equation}
\end{lemma}
\begin{proof}
	We show the identity by induction over $T$.
	For $T=1$,  equation  \eqref{e:sumeq} is easy to verify and subsequently, we
	%\begin{equation*}
	%a + \frac{ba}{a-b}
	%= \frac{a^{2}}{a-b}.
	%\end{equation*}
	assume that the equality holds for $T-1$ with $T\in \N_{\ge 2}$. Now we obtain 
	\begin{align*}
	a + \sum_{t=1}^{T}\frac{ba^t}{(a-b)^t}
	= 
	\frac{ba^T}{(a-b)^{T}}+a+\sum_{t=1}^{T-1}\frac{ba^t}{(a-b)^t}\\
	=
	\frac{ba^T}{(a-b)^{T}}+
	\frac{a^{T}}{(a-b)^{T-1}}
	=
	\frac{ba^T+a^T(a-b)}{(a-b)^{T}}=
	\frac{a^{T+1}}{(a-b)^{T}}
	\end{align*}
	which completes the proof.
\end{proof}

\end{document}